\documentclass{birkjour}
 \newtheorem{thm}{Theorem}[section]
 
 \newtheorem{lem}[thm]{Lemma}
 \newtheorem{prop}[thm]{Proposition}
 \theoremstyle{definition}
 \newtheorem{defn}[thm]{Definition}
 \theoremstyle{remark}

 \numberwithin{equation}{section}

\usepackage{MnSymbol}
\usepackage{stmaryrd}
\usepackage{graphicx}
\usepackage[all]{xy}
\usepackage{color}
\usepackage{hyperref}

\newcommand\ie{{i.e.},\;}

\newcommand{\ga}{\gamma}
\newcommand{\Ga}{\Gamma}

\newcommand\ld{\lambda}

\newcommand\De{\Delta}

\newcommand{\map}{\rightarrow}               

\newcommand{\sa}{{\rm sa}}

\newcommand{\hA}{{\hat A}}
\newcommand\hB{\hat B}

\newcommand{\hP}{{\hat P}}

\newcommand{\Hi}{\mathcal{H}}

\newcommand\BH{\mathcal{B(H)}}

\newcommand\eq[1]{(\ref{#1})}

\newcommand\Ain[1]{A\,\varepsilon\,#1}

\newcommand\daso{\delta^o}
\newcommand\dasi{\delta^i}

\newcommand\dasB[1]{\breve{\delta}(#1)}

\renewcommand\sp[1]{{\rm sp}(\hat A)}

\newcommand\ps[1]{\underline{#1}}        
\newcommand{\dG}{\ps{\mkern1mu\raise2.5pt\hbox{$\scriptscriptstyle|$}
        {\mkern-7mu\rm O}}}                 
\newcommand{\dOU}{\ps{\mkern1mu\raise2.5pt\hbox{$\scriptscriptstyle|$}
        {\mkern-7mu\rm U}}}               

\newcommand{\Sig}{\ps{\Sigma}}            
\newcommand{\R}{\mathcal{R}}                 

\newcommand\Subcl[1]{{\rm Sub}_{{\rm cl}}(#1)} 

\newcommand\Set{{\bf Set}}                    
\newcommand\SetC[1]{\Set^{{#1}^{\rm op}}}      

\newcommand\N{\mathcal{N}}

\newcommand\SetVNop{\SetC{\mathcal{V(N)}}}		

\newcommand\VN{\mathcal{V(N)}}

\renewcommand\O[1]{\mathcal{O}(#1)}						

\newcommand\bbC{\mathbb{C}}
\newcommand\bbR{\mathbb{R}}
\newcommand\bbN{\mathbb{N}}
\newcommand\PN{\mathcal{P}(\N)}

\newcommand\mc[1]{\mathcal{#1}}

\newcommand\Rlr{\ps{\bbR^\leftrightarrow}}

\newcommand\da[1]{\downarrow\!\!{#1}}
\newcommand\ra{\rightarrow}

\newcommand\lra{\longrightarrow}
\newcommand\lmt{\longmapsto}

\def\I{\mathbf{I}}
\def\IR{\I\bbR}

\def\RlrB{\underline{\R^\leftrightarrow_\bot}}
\def\RlrSB{\underline{S\R^\leftrightarrow_\bot}}

\def\R{\mathbb{R}}

\def\Q{\mathbb{Q}}

\def\setdef#1#2{\{#1\mid #2\}}             
\def\enset#1{\mathopen{ \{ }#1\mathclose{ \} }} 
\newcommand{\fdec}[3]{#1: #2 \longrightarrow #3}
\newcommand{\fdef}[3]{#1: #2 \longmapsto #3}

\def\Exists#1{\exists_{#1}\boldsymbol{.}\;}

\def\implies{\Rightarrow}

\def\ps{\mathcal{P}}

\def\L{\mathcal{L}}

\def\pair#1{\left\langle#1\right\rangle}

\def\monic{\xymatrix@C=4ex @R=1.2ex{\, \ar@{>->} [r]  &}\!\!\!}
\def\lmonic#1{\xymatrix@C=4ex @R=1.2ex{\, \ar@{>->}^{#1} [r]  &}\!\!\!}
\def\epic{\!\!\xymatrix@C=3.7ex @R=1.2ex{\ar@{->>}[r]  &}\!\!}
\def\lepic#1{\!\!\xymatrix@C=3.7ex @R=1.2ex{\ar@{->>}^{#1}[r]  &}\!\!}
\def\inclusion{\!\xymatrix@C=3.9ex @R=1.2ex{\ar@{^{(}->} [r]  &}\!\!}
\def\larrow#1{\!\xymatrix@C=3.9ex @R=1.2ex{\ar@{->}^{#1} [r]  &}\!\!}

\def\Sig{\underline{\Sigma}}

\def\O{\mathcal{O}}

\def\sm{\sqsubseteq}
\def\lar{\sqsupseteq}
\def\waybelow{<\!\!<}

\def\osup{\bigsqcup\!}

\def\dsup{\bigsqcup~\!\!^\uparrow}

\def\up#1{\uparrow\!#1}
\def\down#1{\downarrow\!#1}
\def\uup#1{\twoheaduparrow\!#1}
\def\ddown#1{\twoheaddownarrow\!#1}

\def\Top{\mathbf{\mathsf{Top}}}
\def\Pos{\mathbf{\mathsf{Pos}}}

\def\I{\mathbf{I}}
\def\IR{\I\R}
\def\IQ{\I\Q}

\def\one{\mathbf{1}}

\def\Rlr{\underline{\R^\leftrightarrow}}
\def\RlrS{\underline{S\R^\leftrightarrow}}

\def\aA{\mathcal{A}}

\def\aS{\mathcal{S}}

\def\aF{\mathcal{F}}

\def\Cc{\mathbb{C}}

\def\wb{\waybelow}
\def\rmean#1{\mathopen{|}#1\mathclose{|}}

\definecolor{darkgreen}{rgb}{0,.66,0}

\newcommand\comp{\cdot}

\usepackage{color}
\definecolor{orange}{rgb}{1,0.5,0}

\begin{document}

%
%
%
%
%
%
%
%
%

\title[Unsharp Values, Domains and Topoi]
 {Unsharp Values, Domains and Topoi}

\author[Andreas D\"oring]{Andreas D\"oring}

\address{%
Quantum Group\\
Department of Computer Science\\
University of Oxford\\
Wolfson Building\\
Parks Road\\
Oxford OX1 3QD, UK}

\email{andreas.doering@cs.ox.ac.uk}


\author[Rui Soares Barbosa]{Rui Soares Barbosa}

\address{%
Quantum Group\\
Department of Computer Science\\
University of Oxford\\
Wolfson Building\\
Parks Road\\
Oxford OX1 3QD, UK}
\email{rui.soaresbarbosa@wolfson.ox.ac.uk}



\date{January 31, 2011}

\begin{abstract}
  The so-called topos approach provides a radical reformulation of quantum theory. Structurally, quantum theory in the topos formulation is very similar to classical physics. There is a state object $\Sig$, analogous to the state space of a classical system, and a quantity-value object $\Rlr$, generalising the real numbers. Physical quantities are maps from the state object to the quantity-value object -- hence the `values' of physical quantities are not just real numbers in this formalism. Rather, they are families of real intervals, interpreted as `unsharp values'. We will motivate and explain these aspects of the topos approach and show that the structure of the quantity-value object $\Rlr$ can be analysed using tools from domain theory, a branch of order theory that originated in theoretical computer science. Moreover, the base category of the topos associated with a quantum system turns out to be a domain if the underlying von Neumann algebra is a matrix algebra. For general algebras, the base category still is a highly structured poset. This gives a connection between the topos approach, noncommutative operator algebras and domain theory. In an outlook, we present some early ideas on how domains may become useful in the search for new models of (quantum) space and space-time.
\end{abstract}

\maketitle
\begin{flushright}
	\textit{``You cannot depend on your eyes\\
	 when your imagination is out of focus.''}\\
	 Mark Twain (1835--1910)
\end{flushright}
\section{Introduction}
The search for a theory of quantum gravity is ongoing. There is a range of approaches, all of them differing significantly in scope and technical content. Of course, this is suitable for such a difficult field of enquiry. Most approaches accept the Hilbert space formalism of quantum theory and try to find extensions, additional structures that would capture gravitational aspects, or reproduce them from the behaviour of underlying, more fundamental entities.

Yet, standard quantum theory is plagued with many conceptual difficulties itself. Arguably, these problems get more severe when we try to take the step to quantum gravity and quantum cosmology. For example, in the standard interpretations of quantum theory, measurements on a quantum system by an external classical observer play a central role. This concept clearly becomes meaningless if the whole universe is to be treated as a quantum system. Moreover, standard quantum theory and quantum field theory are based on a continuum picture of space-time. Mathematically, the continuum in the form of the real numbers underlies all structures like Hilbert spaces, operators, manifolds, and differential forms (and also strings and loops). If space-time fundamentally is not a smooth continuum, then it may be wrong to base our mathematical formalism on the mathematical continuum of the real numbers.

These considerations motivated the development of the \emph{topos approach} to the formulation of physical theories, and in particular the topos approach to quantum theory. In a radical reformulation based upon structures in suitable, physically motivated topoi, all aspects of quantum theory -- states, physical quantities, propositions, quantum logic, etc. -- are described in a novel way. As one aspect of the picture emerging, physical quantities take their values not simply in the real numbers. Rather, the formalism allows to describe `unsharp', generalised values in a systematic way. In this article, we will show that the mathematical structures used to formalise unsharp values can be analysed using techniques from domain theory. We will take some first steps connecting the topos approach with domain theory.

Domain theory is a branch of order theory and originated in theoretical computer science, where it has become an important tool. Since domain theory is not well-known among physicists, we will present all necessary background here. Recently, domain theory has found some applications in quantum theory in the work of Coecke and Martin \cite{CM11} and in general relativity in the work of Martin and Panangaden \cite{MP06}. Domain theory also has been connected with topos theory before, in the form of synthetic domain theory, but this is technically and conceptually very different from our specific application.

The plan of the paper is as follows: in section \ref{Sec_ToposApproach}, we will present a sketch of the topos approach to quantum theory, with some emphasis on how generalised, `unsharp' values for physical quantities arise. In section \ref{Sec_DomainTh}, we present some background on domain theory. In section \ref{Sec_RlrAsDomain}, it will be shown that the structure of the quantity-value object $\Rlr$, a presheaf whose global elements are the unsharp values, can be analysed with the help of domain-theoretical techniques. Section \ref{Sec_V(N)AsDomain} shows that the base category $\VN$ of the topos $\SetVNop$ associated with a quantum system is a directed complete poset, and moreover an algebraic domain if $\N$ is a matrix algebra. Physically, $\VN$ is the collection of all classical perspectives on the quantum system. In section \ref{Sec_VNNoDomainInInfDim}, we show that the poset $\VN$ is not continuous, and hence not a domain, for non-matrix algebras $\N$, and in section \ref{Sec_Outlook}, we present some speculative ideas on how domains may become useful in the construction of new models of space and space-time adequate for quantum theory and theories `beyond quantum theory' in the context of the topos approach. Section \ref{Sec_Conclusion} concludes. 

\section{The Topos Approach, Contexts and Unsharp Values}			\label{Sec_ToposApproach}
\paragraph{Basic ideas.} In recent years, the \emph{topos approach} to the formulation of physical theories has been developed by one of us (AD), largely in collaboration with Chris Isham \cite{DI(1),DI(2),DI(3),DI(4),Doe07b,DI(Coecke)08,Doe08,Doe09,Doe10}. This approach originates from works by Isham \cite{Ish97} and Isham/Butterfield \cite{IB98,IB99,IB00,IB00b,IB02}. Landsman et al. have presented a closely related scheme for quantum theory \cite{HLS09,CHLS09,HLS09b,HLS11}, developing some aspects topos-internally, and Flori has developed a history version \cite{Flo09}.

The main goal of the topos approach is to provide a framework for the formulation of `neo-realist' physical theories. Such a theory describes a physical system by (i) a state space, or more generally, a state object $\Sigma$, (ii) a quantity-value object $\mc R$, where physical quantities take their values, and (iii) functions, or more generally, arrows $f_A:\Sigma\ra\mc R$ from the state object to the quantity-value object corresponding to physical quantities like position, momentum, energy, spin, etc. Both the state object and the quantity-value object are objects in a \emph{topos}, and the arrows between them representing physical quantities are arrows in the topos. Roughly speaking, a topos is a mathematical structure, more specifically a category, that can be seen as a universe of generalised sets and generalised functions between them. 

Each topos has an \emph{internal logic} that is of intuitionistic type. In fact, one typically has a multivalued, intuitionistic logic and not just two-valued Boolean logic as in the familiar topos $\Set$ of sets and functions. One main aspect of the topos approach is that it makes use of the internal logic of a given topos to provide a logic for a physical system. More specifically, the subobjects of the state object, or a suitable subfamily of these, are the representatives of \emph{propositions} about the values of physical quantities. In the simplest case, one considers propositions of the form ``$\Ain\De$'', which stands for ``the physical quantity $A$ has a value in the (Borel) set $\De$ of real numbers''.

As an example, consider a classical system: the topos is $\Set$, the state object is the usual state space, a symplectic manifold $\mc S$, and a physical quantity $A$ is represented by a function $f_A$ from $\mc S$ to the set of real numbers $\bbR$, which in this case is the quantity-value object. The subset $T\subseteq\mc S$ representing a proposition ``$\Ain\De$'' consists of those states $s\in\mc S$ for which $f_A(s)\in\De$ holds (\ie\;$T=f_A^{-1}(\De)$). If we assume that the function $f_A$ representing the physical quantity $A$ is (at least) measurable, then the set $T$ is a Borel subset of the state space $\mc S$. Hence, in classical physics the representatives of propositions are the Borel subsets of $\mc S$. They form a $\sigma$-complete \emph{Boolean algebra}, which ultimately stems from the fact that the topos $\Set$ in which classical physics is formulated has the familiar two-valued Boolean logic as its internal logic. Of course, we rarely explicitly mention $\Set$ and its internal logic -- it is just the usual mathematical universe in which we formulate our theories. As an underlying structure, it usually goes unnoticed.
\medskip
\paragraph{A topos for quantum theory.} For non-relativistic quantum theory, another topos is being used though. The details are explained elsewhere, see \cite{Doe07b,Doe10} for an introduction to the topos approach and \cite{DI(Coecke)08} for a more detailed description. The main idea is to use \emph{presheaves} over the set $\VN$ of abelian subalgebras of the nonabelian von Neumann algebra $\N$ of physical quantities.\footnote{To be precise, we only consider abelian von Neumann subalgebras $V$ of $\N$ that have the same unit element as $\N$. In the usual presentation of this approach, the trivial algebra $\bbC\hat 1$ is excluded from $\VN$. However, here we will keep the trivial algebra as a bottom element of $\VN$. We will occassionally point out which results depend on $\VN$ having a bottom element.} $\VN$ is partially ordered under inclusion; the topos of presheaves over $\VN$ is denoted as $\SetVNop$. The poset $\VN$, also called the \emph{context category}, is interpreted as the collection of all classical perspectives on the quantum system: each \emph{context} (abelian subalgebra) $V\in\VN$ provides a set of commuting self-adjoint operators, representing compatible physical quantities. The poset $\VN$ keeps track of how these classical perspectives overlap, \ie to which degree they are mutually compatible. Presheaves over $\VN$, which are contravariant functors from $\VN$ to $\Set$, automatically encode this information, too. A presheaf is not a single set, but a `varying set': a family $\underline P=(\underline P_V)_{V\in\VN}$ of sets indexed by elements from $\VN$, together with functions $\underline P(i_{V'V}):\underline P_V\ra\underline P_{V'}$ between the sets whenever there is an inclusion $i_{V'V}:V'\ra V$ in $\VN$.

The state object for quantum theory is the so-called \emph{spectral presheaf} $\Sig$ that is given as follows:
\begin{itemize}
	\item To each abelian subalgebra $V\in\VN$, one assigns the Gel'fand spectrum $\Sig_V$ of the algebra $V$;
	\item to each inclusion $i_{V'V}:V'\rightarrow V$, one assigns the function $\Sig(i_{V'V}):\Sig_V\rightarrow\Sig_{V'}$ that sends each $\ld\in\Sig_V$ to its restriction $\ld|_{V'}\in\Sig_{V'}$.
\end{itemize}
One can show that propositions of the form ``$\Ain\De$'' correspond to so-called \emph{clopen subobjects} of $\Sig$. The set $\Subcl{\Sig}$ of clopen subobjects is the analogue of the set of Borel subsets of the classical state space. Importantly, $\Subcl{\Sig}$ is a complete \emph{Heyting algebra}, which relates to the fact that the internal logic of the presheaf topos $\SetVNop$ is intuitionistic (and not just Boolean). Note that unlike in Birkhoff-von Neumann quantum logic, which is based on the non-distributive lattice $\PN$ of projections in the algebra $\N$, in the topos scheme propositions are represented by elements in a distributive lattice $\Subcl{\Sig}$. This allows to give a better interpretation of this form of quantum logic \cite{DI(Coecke)08,Doe09}. The map from the $\PN$ to $\Subcl{\Sig}$ is called \emph{daseinisation of projections}. Its properties are discussed in detail in \cite{Doe10}.
\medskip
\paragraph{Unsharp values.} In this article, we will mostly focus on the quantity-value object for quantum theory and its properties. Like the state object $\Sig$, the quantity-value object, which will be denoted $\Rlr$, is an object in the presheaf topos $\SetVNop$. The topos approach aims to provide models of quantum systems that can be interpreted as realist (or as we like to call them, neo-realist, because of the richer intuitionistic logic coming from the topos). One aspect is that physical quantities should have values at all times, independent of measurements. Of course, this immediately meets with difficulties: in quantum theory, we cannot expect physical quantities to have sharp, definite values. In fact, the Kochen-Specker theorem shows that under weak and natural assumptions, there is no map from the self-adjoint operators to the real numbers that could be seen as an assignment of values to the physical quantities represented by the operators.\footnote{The conditions are (a) each self-adjoint operator $\hA$ is assigned an element of its spectrum and (b) if $\hB=f(\hA)$ for two self-adjoint operators $\hA,\hB$ and a Borel function $f$, then the value $v(\hB)=v(f(\hA))$ assigned to $\hB$ is $f(v(\hA))$, where $v(\hA)$ is the value assigned to $\hA$.} The Kochen-Specker theorem holds for von Neumann algebras \cite{Doe05}.

The simple idea is to use \emph{intervals}, interpreted as `unsharp values', instead of sharp real numbers. The possible (generalised) values of a physical quantity $A$ are real intervals, or more precisely, real intervals intersected with the spectrum of the self-adjoint operator $\hA$ representing $A$. In our topos approach, each self-adjoint operator $\hA$ in the algebra $\N$ of physical quantities is mapped to an arrow $\dasB{\hA}$ from the state object $\Sig$ to the quantity-value object $\Rlr$. (The latter object will be defined below.)

We will not give the details of the construction of the arrow $\dasB{\hA}$ (see \cite{DI(Coecke)08,Doe10}), but we present some physical motivation here. For this, consider two contexts $V,V'\in\VN$ such that $V'\subset V$, that is, $V'$ is a smaller context than $V$. `Smaller' means that there are fewer physical quantities available from the classical perspective described by $V'$ than from $V$, hence $V'$ gives a more limited access to the quantum system. The step from $V$ to $V'\subset V$ is interpreted as a process of coarse-graining.

For example, we may be interested in the value of a physical quantity $A$. Let us assume that the corresponding self-adjoint operator $\hA$ is contained in a context $V$, but not in a context $V'\subset V$. For simplicity, let us assume furthermore that the state of the quantum system is an eigenstate of $\hA$. Then, from the perspective of $V$, we will get a sharp value, namely the eigenvalue of $\hA$ corresponding to the eigenstate. But from the perspective of $V'$, the operator $\hA$ is not available, so we have to approximate $\hA$ by self-adjoint operators in $V'$. One uses one approximation from below and one from above, both taken with respect to the so-called spectral order. These approximations always exist. In this way, we obtain two operators $\dasi(\hA)_{V'},\daso(\hA)_{V'}$ in $V'$ which, intuitively speaking, contain as much information about $\hA$ as is available from the more limited classical perspective $V'$. If we now ask for the value of $A$ in the given state from the perspective of $V'$, then we will get two real numbers, one from $\dasi(\hA)_{V'}$ and one from $\daso(\hA)_{V'}$. By the properties of the spectral order, these two numbers lie in the spectrum of $\hA$. We interpret them as the endpoints of a real interval, which is an `unsharp value' for the physical quantity $A$ from the perspective of $V'$. Note that we get `unsharp values' for each context $V\in\VN$ (and for some $V$, we may get sharp values, namely in an eigenstate-eigenvalue situation).

In a nutshell, this describes the idea behind \emph{daseinisation of self-adjoint operators}, which is a map from self-adjoint operators in the nonabelian von Neumann algebra $\N$ of physical quantities to arrows in the presheaf topos $\SetVNop$, sending $\hA\in\N_{\sa}$ to $\dasB{\hA}\in{\rm Hom}(\Sig,\Rlr)$. The arrow $\dasB{\hA}$ is the topos representative of the physical quantity $A$. We will now consider the construction of the quantity-value object $\Rlr$, the codomain of $\dasB{\hA}$.

As a first step, we formalise the idea of unsharp values as real intervals. Define
\begin{equation}			\label{Def_IR}
				\IR:=\{[a,b] \mid a,b\in\bbR,\;a\leq b\}.
\end{equation}
Note that we consider closed intervals and that the case $a=b$ is included, which means that the intervals of the form $[a,a]$ are contained in $\IR$. Clearly, these intervals can be identified with the real numbers. In this sense, $\bbR\subset\IR$.

It is useful to think of the presheaf $\Rlr$ as being given by one copy of $\IR$ for each classical perspective $V\in\VN$. Each observer hence has the whole collection of `unsharp' values available. The task is to fit all these copies of $\IR$ together into a presheaf. In particular, whenever we have $V,V'\in\VN$ such that $V'\subset V$, then we need a function from $\IR_V$ to $\IR_{V'}$ (here, we put an index on each copy of $\IR$ to show to which context the copy belongs). The simplest idea is to send each interval $[a,b]\in\IR_V$ to the same interval in $\IR_{V'}$. But, as we saw, in the topos approach the step from the larger context $V$ to the smaller context $V'$ is seen as a process of coarse-graining. Related to that, we expect to get an even more unsharp value, corresponding to a bigger interval, in $\IR_{V'}$ than in $\IR_V$ in general. In fact, we want to be flexible and define a presheaf such that we can map $[a,b]\in\IR_V$ either to the same interval in $\IR_{V'}$, or to any bigger interval $[c,d]\supset[a,b]$, depending on what is required.

We note that as we go from a larger context $V$ to smaller contexts $V'\subset V,\;V''\subset V', ...$, the left endpoints of the intervals will get smaller and smaller, and the right endpoints will get larger and larger in general. The idea is to formalise this by two functions $\mu_V,\nu_V:\da V\ra\bbR$ that give the left resp. right endpoints of the intervals for all $V'\subseteq V$. Here, $\da V:=\{V'\in\VN \mid V'\subseteq V\}$ denotes the downset of $V$ in $\VN$. Physically, $\da V$ is the collection of all subcontexts of $V$, that is, all smaller classical perspectives than $V$. This leads to the following definition.

\begin{defn}
The \emph{quantity-value object $\Rlr$} for quantum theory is given as follows:
\begin{itemize}
	\item To each $V\in\VN$, we assign the set
	\begin{align}			\nonumber
				\Rlr_V:=\{ &(\mu,\nu) \mid \mu,\nu:\da V\ra\bbR,\\
				&\mu\text{ order-preserving},\;\nu\text{ order-reversing},\;\mu\leq\nu\};
	\end{align}
	\item to each inclusion $i_{V'V}:V'\ra V$, we assign the map
	\begin{align}
				\Rlr(i_{V'V}):\Rlr_V &\longrightarrow \Rlr_{V'}\\			\nonumber
				(\mu,\nu) &\longmapsto (\mu|_{\downarrow V'},\nu|_{\downarrow V'}).
	\end{align}
\end{itemize}
\end{defn}
A \emph{global element $\ga$} of the presheaf $\Rlr$ is a choice of one pair of functions $\ga_V=(\mu_V,\nu_V)$ for every context $V\in\VN$ such that, whenever $V'\subset V$, one has $\ga_{V'}=(\mu_{V'},\nu_{V'})=(\mu_V|_{V'},\nu_V|_{V'})=\ga_V|_{V'}$. Clearly, a global element $\ga$ gives a pair of functions $\mu,\nu:\VN\ra\bbR$ such that $\mu$ is order-preserving (smaller contexts are assigned smaller real numbers) and $\nu$ is order-reversing (smaller contexts are assigned larger numbers). Note that $\mu$ and $\nu$ are defined on the whole poset $\VN$. Conversely, each such pair of functions determines a global element of $\Rlr$. Hence we can identify a global element $\ga$ with the corresponding pair of functions $(\mu,\nu)$. We see that $\ga=(\mu,\nu)$ provides one closed interval $[\mu(V),\nu(V)]$ for each context $V$. Moreover, whenever $V'\subset V$, we have $[\mu(V'),\nu(V')]\supseteq[\mu(V),\nu(V)]$, that is, the interval at $V'$ is larger than or equal to the interval at $V$. We regard a global element $\ga$ of $\Rlr$ as one unsharp value. Each interval $[\mu(V),\nu(V)]$, $V\in\VN$, is one component of such an unsharp value, associated with a classical perspective/context $V$. The set of global elements of $\Rlr$ is denoted as $\Ga\Rlr$.

In the following, we will show that the set $\Ga\Rlr$ of unsharp values for physical quantities that we obtain from the topos approach is a highly structured poset, and that a subset of $\Ga\Rlr$ naturally can be seen as a so-called \emph{domain} if the context category $\VN$ is a domain (see section \ref{Sec_RlrAsDomain}). Domains play an important role in theoretical computer science. The context category $\VN$ turns out to be a domain, even an algebraic domain, if $\N$ is a matrix algebra. This leads to a first, simple connection between noncommutative operator algebras, the topos approach and domain theory (see section \ref{Sec_V(N)AsDomain}). For more general von Neumann algebras $\N$, $\VN$ is a not a domain, see Section \ref{Sec_VNNoDomainInInfDim}.

\section{Domain theory}			\label{Sec_DomainTh}
\paragraph{Basics.} This section presents some basic concepts of domain theory. Standard references are \cite{GHK03,AJ94}. Since domain theory is not well-known among physicists, we give some definitions and motivation. Of course, we barely scratch the surface of this theory here.

The study of domains was initiated by Dana Scott \cite{Sco70,Sco72}, with the aim of finding a
denotational semantics for the untyped $\lambda$-calculus. Since then, it has undergone significant development and has become a mathematical theory in its own right, as well as an important tool in theoretical computer science.

Domain theory is a branch of order theory, yet with a strong topological flavour, as it captures notions of convergence and approximation. The basic concepts are easy to grasp: the idea is to regard a partially ordered set as a (qualitative) hierarchy of information content or knowledge. Under this interpretation, we think of $x \sm y$ as meaning `$y$ is more specific, carries more information than $x$'. Therefore, non-maximal elements in the poset represent incomplete/partial knowledge, while maximal elements represent complete knowledge. In a more computational perspective, we can see the non-maximal elements as intermediate results of a computation that proceeds towards calculating some maximal element (or at least a larger element with respect to the information order).

For the rest of this section, we will mainly be considering a single poset, which we will denote by $\pair{P,\sm}$.
\medskip
\paragraph{Convergence -- directed completeness of posets.} The first important concept in domain theory is that of \emph{convergence}. We start by considering some special subsets of $P$.

\begin{defn}
A nonempty subset $S \subseteq P$ is \emph{directed} if 
\begin{equation}
			\forall x,y \in S \exists z \in S: x, y \sm z.
\end{equation}
\end{defn}
A directed set can be seen as a consistent specification of information: the existence of a $z \lar x,y$ expresses that $x$ and $y$ are compatible, in the sense that it is possible to find an element which is larger, \ie contains more information, than both $x$ and $y$. Alternatively, from the computational viewpoint, directed sets describe computations that converge in the sense that for any pair of intermediate results (that can be reached in a finite number of steps), there exists a better joint approximation (that can also be reached in a finite number of steps). This is conceptually akin to converging sequences in a metric space. Hence the natural thing to ask of directed sets is that they possess a suitable kind of limit -- that is, an element containing all the information about the elements of the set, but not more information. This limit, if it exists, can be seen as the ideal result the computation approximates. This leads to the concept of directed-completeness:

\begin{defn}
A \emph{directed-complete poset} (or \emph{dcpo}) is a poset in which any directed set has a supremum (least upper bound). 
\end{defn}
Following \cite{AJ94}, we shall write $\dsup S$  to denote the supremum of a directed
set $S$, instead of simply $\osup S$. Note that $\dsup S$ means `S is a directed set, and we are considering its supremum'.


The definition of morphisms between dcpos is the evident one:
\begin{defn}
A function $\fdec{f}{P}{P'}$ between dcpos $\pair{P, \sm}$ and $\pair{P',\sm'}$ is
\emph{Scott-continuous} if
\begin{itemize}
\item $f$ is order-preserving (monotone);
\item for any directed set $S \subseteq P$, $f(\dsup S) = \dsup f^\ra(S)$, where $f^\ra(S)=\{f(s) \mid s\in S\}$.
\end{itemize}
\end{defn}
Clearly, dcpos with Scott-continuous functions form a category. The definition of a Scott-continuous function can be extended to posets which are not dcpos, by carefully modifying the second condition to say `for any directed set $S$ that has a supremum'. The reference to `continuity' is not fortuitous, as there is the so-called \emph{Scott topology}, with respect to which these (and only these) arrows are continuous. The Scott topology will be defined below.
\medskip
\paragraph{Approximation - continuous posets.} The other central notion is sometimes called approximation. This is captured by the following relation on elements of $P$.
\begin{defn}
We say that \emph{$x$ approximates $y$} or \emph{$x$ is way below $y$}, and write
$x \waybelow y$, whenever for any directed set $S$ with a supremum, we have
\begin{equation}
			y \sm \dsup S \implies \exists s \in S: x \sm s.
\end{equation}
\end{defn}
The `way-below' relation captures the fact that $x$ is much simpler than $y$, yet carries essential information about $y$. In the computation analogy, we could say that $x$ is an unavoidable step in any computation of $y$, in the sense that any computation that tends to (i.e., successively approximates) $y$ must reach or pass $x$ in a finite amount of steps.

In particular, one can identify certain elements which are `finite' or `simple', in the sense that they cannot be described by (\ie given as the supremum of) any set of smaller elements that does not contain the element itself already. 
\begin{defn}
An element $x \in P$ such that $x \waybelow x$ is called a
\emph{compact} or a \emph{finite} element.
$K(P)$ stands for the set of compact elements of $P$.
\end{defn}
Another interpretation of a compact element $x$ is to say
that any computation that tends to $x$ eventually reaches $x$ in a finite number of steps.

In a poset $\pair{\ps A, \subseteq}$ of subsets of a set $A$, the compact elements are exactly the finite subsets of $A$:
if one covers a finite set $F$ by a directed collection $(S_i)_{i\in I}$, $F$ will be contained in one of the $S_i$
already.
Also, the definition of the way-below relation (particularly of $x \waybelow x$) has a striking similarity with that of a compact set in topology. Indeed, in the poset $\pair{\O(X),\subseteq}$ of open subsets of a topological space $X$, the compact elements
are simply the compact open sets.

Given an element $x$ in a poset $P$, we write $\down x$ for the \emph{downset} $\{y\in P \mid y\leq x\}$ of $x$ in $P$. If $X\subseteq P$, then $\down X:=\{y\in P \mid \exists x\in X:y\leq x\}$. The sets $\up x$ and $\up X$ are defined analogously. Similarly, we write $\ddown x, \ddown X, \uup x, \uup X$ for the corresponding sets with respect to the way-below relation $\waybelow$, e.g.
\begin{equation}
			\ddown x \;\;\;:=\;\;\; \setdef{y\in P}{y\waybelow x}.
\end{equation}

We now come to another requirement that is usually imposed on the posets of interest, besides directed-completeness.
\begin{defn}
A poset $P$ is a \emph{continuous poset} if, for any $y \in P$, one has $\dsup\ddown y = y$, and $P$ is an \emph{algebraic poset} if, for any $y \in P$, $\dsup(\ddown y \cap K(P)) = y$ holds.
\end{defn}
Recall that $\dsup\ddown y = y$ means that $\ddown y$ is directed and has supremum $y$. 
Continuity basically says that the elements `much simpler' than $y$ carry all the  
information about $y$, when taken together. Algebraicity further says that the `primitive' (\ie compact) elements are enough.

\paragraph{Bases and domains.} The continuity and algebraicity requirements are often expressed in terms of the notion of a basis. The definition is slightly more involved, but the concept of a basis is useful in its own right.
\begin{defn}
A subset $B \subseteq P$ is a \emph{basis} for $P$ 
if, for all $x \in X$, $\dsup(B \cap \ddown x) = x$. 
\end{defn}
It is immediate that continuity implies
that $P$ itself is a basis. Conversely, the existence of a basis implies continuity.
\begin{defn}
A \emph{domain} (or \emph{continuous domain}) $\pair{D,\sm}$ is a dcpo which is continuous. Equivalently, a domain is a dcpo that has a basis. $\pair{D,\sm}$ is an \emph{algebraic domain} if it is a domain and algebraic, that is, if the set $K(D)$ of compact elements is a basis for $\pair{D,\sm}$. An \emph{$\omega$-continuous} (resp. \emph{$\omega$-algebraic}) domain is a continuous (resp. algebraic) domain with a countable basis.
\end{defn}
Note that a domain always captures the notions of convergence and of approximation as explained above.

\medskip
\paragraph{Bounded complete posets.} Later on, we will need another completeness property that a poset $P$ may or may not have.
\begin{defn}
  A poset is \emph{bounded complete} (or a \emph{bc-poset}) if all subsets $S$ with an upper bound have a supremum.
  It is \emph{finitely bounded complete} if all finite subsets with an upper bound have a supremum.
  It is \emph{almost (finitely) bounded complete} if all nonempty (finite) subsets with an upper bound have a supremum.\footnote{Note that the `almost' versions don't require a least element $\bot$.}
\end{defn}
We state the following result without proof.
\begin{prop}			\label{Prop_FinBC-dcpoBoundedComplete}
A(n almost) finitely bounded complete dcpo is (almost) bounded complete.\footnote{Bounded-complete dcpos are the same as complete semilattices (see \cite{GHK03}).}
\end{prop}

Another property, stronger than bounded completeness, will be needed later on:
\begin{defn}
  An \emph{$L$-domain} is a domain $D$ in which, for each $x \in D$, the principal ideal $\down x$ is a complete lattice.
\end{defn}
\medskip
\paragraph{The Scott topology.} We now define the appropriate topology on dcpos and domains, called the Scott toplogy, and present some useful results. In fact, the Scott topology can be defined on any poset, but we are mostly interested in dcpos and domains.
\begin{defn}
Let $\pair{P,\leq}$ be a poset. A subset $G$ of $P$ is said to be \emph{Scott-open} if
\begin{itemize}
	\item $G$ is an upper set, that is 	
	\begin{equation}
				x \in G \land x \leq y \implies y \in G;
	\end{equation}
\item $G$ is inaccessible by directed suprema, i.e. for any directed set $S$ with a supremum,
	\begin{equation}
		\dsup S \in G \implies \Exists{s \in S} s \in U.
	\end{equation}
\end{itemize}
The complement of a Scott-open set is called a \emph{Scott-closed} set. Concretely, this is a lower set closed for all existing directed suprema.
\end{defn}
The name Scott-open is justified by the following result.
\begin{prop}\label{Scott-top}
The Scott-open subsets of $P$ are the opens of a topology on $P$, called the \emph{Scott topology}.
\end{prop}
\begin{prop}\label{Scott-basis}
If $P$ is a continuous poset, the collection
\begin{equation}
			\setdef{\uup x}{x \in P}
\end{equation}
is a basis for the Scott topology.
\end{prop}
The Scott topology encodes a lot of information about the domain-theoretical properties of $P$ relating to convergence and (in the case of a continuous poset) approximation. The following is one of its most important properties, relating the algebraic and topological aspects of domain theory.

\begin{prop}\label{Scott-cts-is-cts}
Let $P$ and $Q$ be two posets. A function $\fdec{f}{P}{Q}$ is Scott-continuous if and only if it is (topologically) continuous with respect to the Scott topologies on $P$ and $Q$.
\end{prop}
The results above (and proofs for the more involved ones) can be found in \cite[\textsection 1.2.3]{AJ94}. More advanced results can be found in \cite[\textsection 4.2.3]{AJ94}.

As for separation properties, the Scott topology satisfies only a very weak axiom in all interesting cases.
\begin{prop}\label{scott-T0}
The Scott topology on $P$ gives a $T0$ topological space. It is $T2$ if and only if the order in $P$ is trivial.
\end{prop}
\medskip
\paragraph{The real interval domain.} As we saw in section \ref{Sec_ToposApproach}, the collection of real intervals can serve as a model for `unsharp values' of physical quantities (at least if we consider only one classical perspective $V$ on a quantum system). We will now see that the set $\IR$ of closed real intervals defined in equation \eq{Def_IR} actually is a domain, the so-called \emph{interval domain}. This domain was introduced by Scott \cite{Sco72b} as a computational model for the real numbers. 

\begin{defn}
The \emph{interval domain} is the poset of closed intervals in $\R$ (partially) ordered by reverse inclusion,
\begin{equation}
			\IR := \pair{\{ [a,b] \mid a, b  \in \R\} , \sm := \supseteq}.
\end{equation}
\end{defn}
The intervals are interpreted as approximations to real numbers, hence the ordering
by reverse inclusion: we think of $x \sm y$ as `$y$ is sharper than $x$'. Clearly, the maximal elements are the real numbers themselves (or rather, more precisely, the intervals of the form $[a,a]$). 

We shall denote by $x_-$ and $x_+$  the left and right endpoints of an interval $x \in \IR$.
That is, we write $x = [x_-,x_+]$. Also, if $\fdec{f}{X}{\IR}$ is a function to the interval
domain, we define the functions
$\fdec{f_-, f_+}{X}{\R}$ given by $f_\pm(x) := (f(x))_\pm$,
so that, for any $x \in X$, $f(x) = [f_-(x),f_+(x)]$.
Clearly, one always has $f_- \leq f_+$ (the order on
functions being defined pointwise).
Conversely, any two functions $\fdec{g,h}{X}{\R}$ with $g \leq h$
determine a function $f$ such that $f_-=g$ and $f_+=h$.

Writing $x = [x_-,x_+]$ amounts to regarding $\IR$ as being embedded in $\R \times \R$ (as a set).
The decomposition for functions can then be depicted as follows:
\begin{equation}
  \xymatrix{
  & & & \R &&\\
  X \ar^{\;\;f}[rr] \ar^{f_-}[rrru] \ar_{f_+}[rrrd]  & & \IR \; \ar@{>->}[r] & \R \times \R \ar^{\pi_1}[d] \ar_{\pi_2}[u]& & (f_- \leq f_+)\\ 
  & & & \R &&\\
  }
  \label{diag}
\end{equation}
and it is nothing more than the universal property of the (categorical) product $\R \times \R$ restricted to $\IR$, the restriction being reflected
in the condition $f_- \leq f_+$.

This diagram is more useful in understanding $\IR$ than it may seem at first sight.
Note that we can place this diagram in the category $\Pos$ (of posets and monotone maps)
if we make a judicious choice
of order in $\R^2$. This is achieved by equipping the first copy of $\R$ with its usual order 
and the second copy with the opposite order $\geq$.
Adopting this view, equations \ref{eq:waybelowIR} and \ref{eq:dsupIR} below
should become apparent.
\medskip
\paragraph{Domain-theoretic structure on $\IR$.} The way-below relation in $\IR$ is given by:
\begin{equation}
  		x \waybelow y \;\;\;\text{iff}\;\;\; (x_- < y_-) \land (y_+ < x_-).			\label{eq:waybelowIR}
\end{equation}
Suprema of directed sets exist and are given by intersection. This can be written directly
as an interval: let $S$ be a directed set, then
\begin{equation}
  		\dsup S \;=\; \bigcap S \;=\; [\sup\setdef{x_-}{x\in S}, \inf\setdef{x_+}{x\in S}].			\label{eq:dsupIR}
\end{equation}
One observes easily that $\IR$ is an $\omega$-continuous dcpo and hence a domain. (To show $\omega$-continuity, one can consider the basis $\IQ$ for $\IR$.)

Moreover, $\IR$ is a meet-semilattice. Also, we observe that $\IR$ is an almost-bounded-complete poset: if $S\subset\IR$ is a non-empty subset with an upper bound (which just means that all intervals in $S$ overlap), then $S$ has a supremum,  clearly given by the intersection of the intervals. The related poset $\IR_{\bot}$ (where we add a least element, which can be interpreted as $\bot = \R = [-\infty,+\infty]$) is then bounded complete. Also, it is easy to see that $\IR_\bot$ is an $L$-domain.

We now consider the Scott topology on $\IR$. The basic open sets of the Scott topology on $\IR$ are of the form
\begin{equation}
			\uup [a,b] = \setdef{[c,d]}{a < c \leq d < b}
\end{equation}
for each $[a,b] \in \IR$. The identity
\begin{equation}
			\uup [a,b] = \setdef{t \in \IR}{t \subseteq (a,b)}
\end{equation}
allows us to see $\uup[a,b]$ as a kind of open interval $(a,b)$. More precisely, it consists of closed intervals contained in the real interval $(a,b)$.

Recall that we can see the poset $\IR$ as sitting inside $\R^\leq \times \R^\geq$.
In topological terms, this means that the Scott topology is inherited from the Scott
topologies in $\R^\leq$ and $\R^\geq$. These are simply the usual lower and upper
semicontinuity topologies on $\R$, with basic
open sets respectively $(a,\infty)$ and $(-\infty,b)$.
Interpreting diagram \ref{diag} in $\Top$ gives the following result.
\begin{prop}\label{res:scott-lowerupper} Let $X$ be a topological space and
$\fdec{f}{X}{\IR}$ be a function.
Then $f$ is continuous
iff $f_-$ is lower semicontinuous and $f_+$ is upper semicontinuous.
\end{prop}
Hence, the subspace topology on $\IR$ inherited from $\R^{LSC} \times \R^{USC}$, where $LSC$ and $USC$ stand for the lower and upper semicontinuity topologies, is the Scott topology.
\medskip
\paragraph{Generalising $\R$.} We want to regard $\IR$ as a generalisation of $\R$. Note that the
set  $\max\IR$ consists of degenerate intervals $[x,x]=\{x\}$. This gives an obvious way of embedding the usual continuum $\R$ in $\IR$. What is more interesting, this is actually a homeomorphism.
\begin{prop}\label{res:recoverR}
$\R \cong \max \IR$ as topological spaces, where $\R$ is equipped with its usual topology and $\max\IR$ with the subspace topology inherited from $\IR$.
\end{prop}
This result is clear from the following observation that identifies basic open sets:
\[\uup[a,b] \cap \max\IR = \setdef{t \in \IR}{t \subseteq (a,b)} \cap \max\IR
= \setdef{\{t\}}{t \in (a,b)} = (a,b)\]

Another (maybe more informative) way of seeing this is by thinking of $\IR$ as sitting inside $\R^{LSC} \times \R^{USC}$. The well-known fact that the topology of $\R$ is given as the join of the two semicontinuity topologies can be stated as follows: the diagonal map $\fdec{diag}{\R}{\R^{LSC}\times \R^{USC}}$ gives an isomorphim between $\R$ and its image $\Delta$. This is because basic opens in $\Delta$ are
\begin{equation}
		  \Delta \cap \left( (a,\infty) \times (-\infty,b) \right) = (a,\infty) \cap (-\infty,b) = (a,b).
\end{equation}
The result above just says that this diagonal map factors through $\IR$, with image $\max\IR$.

The only separation axiom satisfied by the topology of $\IR$ is the T0 axiom, whereas $\R$ is a T6 space.
However, this topology still keeps some properties of the topology on $\R$. It is second-countable and locally compact  (the general definition of local compactness for non-Hausdorff spaces is given in \cite{Wil70}). Note that, obviously, adding a least element $\bot = [-\infty,+\infty]$
to the domain would make it compact.

\section{The Quantity-Value Object and Domain Theory}			\label{Sec_RlrAsDomain}
We now study the presheaf $\Rlr$ of (generalised) values that shows up in the topos approach in the light of domain theory. First of all, one can consider each component $\Rlr_V$ individually, which is a set. More importantly, we are interested in the set $\Ga\Rlr$ of global elements of $\Rlr$. We will relate these two sets with the interval domain $\IR$ introduced in the previous section.

The authors of \cite{DI(2)}, where the presheaf $\Rlr$ was first introduced as the quantity-value object for quantum theory, were unaware of domain theory at the time. The idea that $\Rlr$ (or rather, a closely related co-presheaf) is related to the interval domain was first presented by Landsman et al. in \cite{HLS09}. These authors considered $\Rlr$ as a topos-internal version of the interval domain. Here, we will focus on topos-external arguments.
\medskip
\paragraph{Rewriting the definition of $\Rlr$.} We start off by slightly rewriting the definition of the presheaf $\Rlr$.

Recall from the previous section that a function $\fdec{f}{X}{\IR}$ to the interval domain can be decomposed into two functions $\fdec{f_-, f_+}{X}{\R}$ giving the left and right endpoints of intervals. In case $X$ is a poset, it is immediate from the definitions that such an $f$ is order-preserving if and only if $f_-$ is order-preserving and $f_+$ is order-reversing with respect to the usual order on the real numbers. This allows us to rewrite the definition of $\Rlr$:
\begin{defn}
The quantity-value presheaf $\Rlr$ is given as follows:
\begin{itemize}
	\item To each $V \in \VN$, we assign the set
	\begin{equation}
				\Rlr_V \;\;\; = \;\;\; \setdef{\fdec{f}{\down V}{\IR}}{\text{$f$ order-preserving}};
	\end{equation}
	\item to each inclusion $i_{V'V}$, we assign the function
	\begin{align*}
		\Rlr(i_{V'V}) : \Rlr_V &\longrightarrow \Rlr_{V'}\\
		f &\longmapsto f|_{\down V'}
	\end{align*}
\end{itemize}
\end{defn}
This formulation of $\Rlr$ brings it closer to the interval domain.
\medskip
\paragraph{Global elements of $\Rlr$.} In section \ref{Sec_ToposApproach}, we stated that in the topos approach, the generalised values of physical quantities are given by global elements of the presheaf $\Rlr$. We remark that the global elements of a presheaf are (analogous to) points if the presheaf is regarded as a generalised set. Yet, the set of global elements may not contain the full information about the presheaf. There are non-trivial presheaves that have no global elements at all -- the spectral presheaf $\Sig$ is an example. In contrast, $\Rlr$ has many global elements.

We give a slightly modified characterisation of the global elements of $\Rlr$ (compare end of section \ref{Sec_ToposApproach}):
\begin{prop}\label{Rlr-globalsections}
Global elements of $\Rlr$ are in bijective correspondence with order-preserving functions from $\VN$ to $\IR$.
\end{prop}
\begin{proof}
Let $\underline\one$ be the terminal object in the topos $\SetVNop$, that is, the constant presheaf with a one-element set $\{*\}$ as component for each $V\in\VN$. A global element of $\Rlr$ is an arrow in $\SetVNop$, \ie a natural transformation
\begin{equation}
			\fdec{\eta}{\underline\one}{\Rlr}.
\end{equation}
For each object $V$ in the base category $\VN$, this gives a function
\begin{equation}
			\fdec{\eta_V}{\{\star\}}{\Rlr_V}
\end{equation}
which selects an element $\ga_V := \eta_V(\star)$ of $\Rlr_V$. Note that each $\ga_V$ is an order-preserving function $\fdec{\ga_V}{\down V}{\IR}$. The naturality condition, expressed by the diagram
\[
\xymatrix{
\one_V    =  \hspace{-1cm} &\{*\} \ar[rr]^{\eta_V} \ar@{=}[dd] & & \Rlr_V \ar[dd]^{\Rlr(i_{V'V})}
\\ &&& \\
\one_{V'} =  \hspace{-1cm} &\{*\} \ar[rr]^{\eta_V} & & \Rlr_{V'} 
}
\]  
then reads $\ga_{V'} = \Rlr(i_{V'V}) (\ga_{V}) = f_V|_{\down V'}$. Thus, a global element of $\Rlr$ determines a unique function
\begin{align*}
			\tilde{\ga}   :  \VN & \longrightarrow \IR\\
			V' &\longmapsto \ga_{V}(V'),
\end{align*}
where $V$ is some context such that $V'\subseteq V$. The function $\tilde\ga$ is well-defined: if we pick another $W\in\VN$ such that $V'\subseteq W$, then the naturality condition guarantees that $\ga_W(V')=\ga_V(V')$. The monotonicity condition for each $\ga_V$ forces $\tilde{\ga}$ to be a monotone (order-preserving function) from $\VN$ to $\IR$.

Conversely, given an order-preserving function $\tilde\ga:\VN\ra\IR$, we obtain a global element of $\Rlr$ by setting
\begin{equation}
			\forall V\in\VN:\ga_V:=\tilde\ga|_{\down V}.
\end{equation}
\end{proof}
\paragraph{Global elements as a dcpo.} So far, we have seen that each $\Rlr_V$, $V\in\VN$, and the set of global elements $\Gamma\Rlr$ are sets of order-preserving functions from certain posets to the interval domain $\IR$. Concretely,
\begin{align}
			&\Rlr_V = \mc{OP}(\down V,\IR),\\
			&\Ga\Rlr = \mc{OP}(\VN,\IR),
\end{align}
where $\mc{OP}(P,\IR)$ denotes the order-preserving functions from the poset $P$ to $\IR$.

We now want to apply the following result (for a proof of a more general result, see Prop. II-4.20 in \cite{GHK03}):
\begin{prop}\label{Prop_C(X,IR)dcpo}
Let $X$ be a topological space. If $P$ is a dcpo (resp. bounded complete dcpo, resp. almost bounded complete dcpo) equipped with the Scott topology, then $C(X,P)$ with the pointwise order is a dcpo (resp. bounded complete dcpo, resp. almost bounded complete dcpo).
\end{prop}

The problem is that if we want to apply this result to our situation, then we need \emph{continuous} functions between the posets, but neither $\down V$ nor $\VN$ have been equipped with a topology so far. (On $\IR$, we consider the Scott topology.) The following result shows what topology to choose:
\begin{prop}
Let $P$, $Q$ be posets, and let $\fdec{f}{P}{Q}$ be a function. If the poset $Q$ is continuous, the following are equivalent:
\begin{enumerate}
	\item\label{monscott-1} $f$ is order-preserving;
	\item\label{monscott-2} $f$ is continuous with respect to the upper Alexandroff topologies on $P$ and $Q$;
	\item\label{monscott-3} $f$ is continuous with respect to the upper Alexandroff topology on $P$ and the Scott topology on $Q$.
\end{enumerate}
\end{prop}
Hence, we put the upper Alexandroff topology on $\down V$ and $\VN$ to obtain the following equalities:
\begin{itemize}
	\item For each $V \in \VN$, we have $\Rlr_V = C((\down V)^{UA},\IR)$;
	\item for the global elements of $\Rlr$, we have $\Gamma\Rlr = C(\VN^{UA},\IR)$.
\end{itemize}
By Prop. \ref{Prop_C(X,IR)dcpo}, since $\IR$ is an almost bounded complete dcpo, both $\Rlr_V$ (for each $V$), and $\Gamma\Rlr$ also are almost bounded complete dcpos.
%
\medskip
\paragraph{A variation of the quantity-value object.} 
The following is a slight variation of the quantity-value presheaf where we allow for completely undetermined values (and not just closed intervals $[a,b]$). This is achieved by including a bottom element in the interval domain $\IR$, this bottom element of course being interpreted as the whole real line. Using the $L$-domain $\IR_\bot$, let $\RlrB$ be the presheaf defined as follows:
\begin{itemize}
	\item To each $V \in \VN$, we assign the set
	\begin{equation}
	\RlrB_V \;\;\; = \;\;\; \setdef{\fdec{f}{\down V}{\IR_\bot}}{\text{$f$ order-preserving}};
	\end{equation}
	\item to each inclusion $i_{V'V}$, we assign the function
	\begin{align*}
		\RlrB(i_{V'V}) : \RlrB_V &\longrightarrow \RlrB_{V'}\\
		f &\longmapsto f|_{\down V'}
	\end{align*}
\end{itemize}
Clearly, we have:
\begin{itemize}
	\item For each $V \in \VN$, $\RlrB_V = C((\down V)^{UA},\IR_\bot)$;
	\item $\Gamma\RlrB = C(\VN^{UA},\IR_\bot)$.
\end{itemize}
Hence, by Prop. \ref{Prop_C(X,IR)dcpo}, $\RlrB_V$ (for each $V$) and $\Ga\Rlr$ are bounded complete dcpos.
Note that $\Rlr$ is a subpresheaf of $\RlrB$.

Also, one can consider the presheaves $\RlrS$ and $\RlrSB$ defined analogously to $\Rlr$ and $\RlrB$, but requiring the functions to be Scott-continuous rather than simply order-preserving. Again, $\RlrS$ is a subpresheaf of $\RlrSB$. Moreover, note that $\RlrS$ (resp. $\RlrSB$) is a subpresheaf of $\Rlr$ (resp. $\RlrB$). For a finite-dimensional $\N$, since the poset $\VN$ has finite height, any order-preserving function is Scott-continuous, hence there is no difference between $\Rlr$ and $\RlrS$ (resp. $\RlrB$ and $\RlrSB$).
The presheaves $\RlrS$ and $\RlrSB$ using Scott-continuous functions are interesting since the arrows of the form $\breve\delta(\hat A):\Sig\ra\Rlr$ that one obtains from daseinisation of self-adjoint operators \cite{DI(2),DI(Coecke)08} actually have image in $\RlrS$. We will not prove this result here, since this would lead us too far from our current interest. We have:
\begin{itemize}
	\item For each $V \in \VN$, $\RlrS_V = C((\down V)^{S},\IR)$ and $\RlrSB_V = C((\down V)^{S},\IR_\bot)$;
	\item $\Gamma\RlrS = C(\VN^{S},\IR)$ and $\Gamma\RlrSB = C(\VN^{S},\IR_\bot)$.
\end{itemize}

\medskip
\paragraph{Domain-theoretic structure on global sections.}
We now consider continuity. The following result (from theorem A in \cite{YJ96}) helps to clarify things further:
\begin{prop}
  If $D$ is a continuous $L$-domain and $X$ a core compact space (i.e. its poset of open sets is continuous), then $C(X,D)$ (where $D$ has the Scott topology) is a continuous $L$-domain.
\end{prop}
In particular, any locally compact space is core compact (\cite{EEK98}). Note that $\VN$ with the Alexandroff topology is always locally compact; it is even compact if we consider $\VN$ to have a least element. Moreover, we saw before that $\IR_\bot$ is an $L$-domain. Hence, we can conclude that the global sections of $\RlrB$ form an $L$-domain. Similarly, the sections of $\RlrB_V$ (over $\down V$) form an $L$-domain.

For the presheaves $\Rlr,\;\RlrS$, and $\RlrSB$, it is still an open question whether their global sections form a domain or not.


\section{The Category of Contexts as a Dcpo}			\label{Sec_V(N)AsDomain}
We now turn our attention to the poset $\VN$. In this section, we investigate it from the perspective of domain theory. We will show that $\VN$ is a dcpo, and that the assignment $\N \mapsto \VN$ gives a functor from von Neumann algebras to the category of dcpos.

From a physical point of view, the fact that $\VN$ is a dcpo shows that the information contained in a coherent set of physical contexts is captured by a larger (limit) context.
If $\N$ is a finite-dimensional algebra, that is, a finite direct sum of matrix algebras, then $\VN$ is an algebraic domain. This easy fact will be shown below. We will show in section \ref{Sec_VNNoDomainInInfDim} that, for other types of von Neumann algebras, $\VN$ is not continuous.

In fact, most of this section is not concerned with $\VN$ itself, but with a more general kind of posets of which $\VN$ is but one example.

The fact that $\VN$ is an algebraic domain in the case of matrix algebras $\N$ was suggested to us in private communication by Chris Heunen.
\medskip
\paragraph{Domains of subalgebras.} Common examples of algebraic domains are the posets of subalgebras of an algebraic structure, e.g. the poset of subgroups of a group. (Actually, this is the origin of the term `algebraic'.) We start by formalising this statement and prove that posets of this kind are indeed domains. This standard result will be the point of departure for the following generalisations to posets of subalgebras modulo equations and topological algebras.

Mathematically, we will be using some simple universal algebra. A cautionary remark for the reader familiar with universal algebra: the definitions of some concepts were simplified. For example, for a fixed algebra $\aA$, its subalgebras are simply defined to be subsets (and not algebras in their own right).

\begin{defn} A \emph{signature} is a set $\Sigma$ of so-called \emph{operation symbols}
  together with a function $\fdec{ar}{\Sigma}{\bbN}$ designating the \emph{arity} of each symbol.
\end{defn}
\begin{defn}
A \emph{$\Sigma$-algebra} (or an algebra with signature $\Sigma$) is a pair
$\aA= \pair{A;\aF}$ consisting of a set $A$ (the \emph{support}) and a set of operations
$\aF = \setdef{f^\aA}{f \in \Sigma}$, where $f^\aA:A^{ar(f)}\longrightarrow A$ is said to realise the operation symbol $f$.
\end{defn}

Note that the signature describes the operations an algebra is required to have. Unless
the distinction is necessary, we will omit the superscript and therefore not make a distinction between operator symbols and operations themselves.

As an example, a monoid $\pair{M; \comp, 1}$ is an algebra with two operations, $\comp$ and $1$, of arities $ar(\comp) = 2$ and $ar(1) = 0$.

\begin{defn}
  Given an algebra $\aA= \pair{A;\aF}$, a \emph{subalgebra} of  $\aA$ is a subset of $A$ which is closed under all operations $f \in \aF$. We denote  the set of subalgebras of $\aA$ by $Sub_{\aA}$.
\end{defn}

\begin{defn}
Let $\aA = \pair{A;\aF}$ be an algebra  and $G \subseteq A$.
The \emph{subalgebra of  $\aA$ generated by $G$}, denoted $\pair{G}$, is the smallest algebra containing $G$. This is given explicitly by the closure of
$G$ under the operations, i.e. given
\begin{align*}
G_0 \; & =\; G,\\
G_{k+1} \; & =\; G_k \cup \bigcup_{f \in \aF} \setdef{f(x_1, \cdots, x_{ar(f)})}{x_1, \cdots, x_{ar(f)} \in G_k},
\end{align*}
we obtain
\begin{equation}
			\pair{G} \; =\; \bigcup_{k \in \bbN} G_k.
\end{equation}
A subalgebra $B \subseteq A$ is said to be \emph{finitely generated} whenever it is generated
by a finite subset of $A$. 
\end{defn}

We will consider the poset $(Sub_{\aA}, \subseteq)$ in some detail. The following results are well-known and rather easy to prove:
\begin{prop}\label{SubAIsCompleteLattice}
$(Sub_{\aA}, \subseteq)$ is a complete lattice with the operations
\begin{align*}
\bigwedge \aS \; & :=\; \bigcap \aS,\\
\bigvee \aS \; & :=\; \pair{\bigcup \aS}
\end{align*}
\end{prop}

\begin{prop}			\label{Prop_UnionOfDirectedSIsSubalg}
If $\aS \subseteq Sub_{\aA}$ is directed, then $\bigvee \aS = \pair{\bigcup \aS} = \bigcup \aS$.
\end{prop}

Let $B \in Sub_{\aA}$. We write
\begin{align}			\nonumber
Sub_{fin}(B)\; & :=\; \setdef{C \in Sub_{\aA}}{C \subseteq B\, \text{and}\, C\, \text{is finitely generated}}\\
                            & =\; \setdef{\pair{x_1, \cdots, x_n}}{n \in \bbN, \enset{x_1, \cdots, x_n} \subseteq B}
\end{align}
for the finitely generated subalgebras of $B$.
\begin{lem}\label{SubAfin}
  For all $B\in Sub_\aA$, we have $B = \dsup Sub_{fin}(B) = \bigcup Sub_{fin}(B)$.
\end{lem}
%

We now characterise the way-below relation for the poset $Sub_\aA$.
\begin{lem}\label{SubAWayBelow}
For $B, C \in Sub_{\mathcal{A}}$, one has $C \wb B$ if and only if $C \subseteq B$ and $C$ is finitely generated.
\end{lem}

\begin{prop}\label{Subdomain}
$Sub_{\aA}$ is an algebraic complete lattice (i.e. a complete lattice which is an algebraic domain).
\end{prop}
\medskip
\paragraph{Domains of subalgebras with additional algebraic properties.} We are interested in subalgebras that satisfy certain algebraic properties that are not present in the algebra $\aA$. For example, if $\aA$ is a monoid, we may be interested in the set of abelian submonoids of $\aA$. To formalise this, we must be able to incorporate 'equational properties'.

\begin{defn}\label{def:poly}
  The set $\Sigma[x_1, \cdots, x_n]$ of terms (or polynomials) over the signature $\Sigma$ in the variables $x_1, \cdots, x_n$ is defined inductively as follows:
\begin{itemize}
  \item For each $i\in\{1,...,n\}$, we have $x_i \in \Sigma[x_1, \cdots, x_n]$;
  \item for each $f \in \Sigma$, if $p_1, \cdots, p_{ar(f)} \in \Sigma[x_1, \cdots, x_n]$, then $f(p_1, \cdots, p_{ar(f)}) \in \Sigma[x_1, \cdots, x_n]$.
\end{itemize}
\end{defn}

Note that $f(p_1, \cdots, p_{ar(f)})$ does not denote the application of the function $f$, but simply a formal string of symbols. The variables are also just symbols.

\begin{defn}\label{def:evaluatepoly}
Let $\aA$ be a $\Sigma$-algebra. Let $p \in \Sigma[x_1, \cdots, x_n]$, and let $\fdec{\nu}{\enset{x_1, \cdots, x_n}}{A}$ (called a valuation of the variables).
Then one extends $\nu$ to a function $\fdec{\rmean{\cdot}_{\nu}}{\Sigma[x_1, \cdots, x_n]}{A}$ by the following inductive rules:
\begin{itemize}
	\item $\rmean{x_i}_{\nu} := \nu(x_i)$;
	\item for each $f \in \Sigma$, $\rmean{f(p_1, \cdots, p_n)}_{\nu} := f^\aA(\rmean{p_1}_{\nu}, \cdots, \rmean{p_n}_{\nu})$.
\end{itemize}
\end{defn}

\begin{defn}
A \emph{polynomial equation} over $\aA$ is a pair $\pair{p,q}$ of polynomials in the same variables. 
A \emph{system of polynomial equations} over $\aA$ is a set of such pairs.
\end{defn}

\begin{defn}
A subalgebra $B \in Sub_{\aA}$ is said to \emph{satisfy an equation}  $\pair{p,q}$ in the variables $\enset{x_1, \cdots, x_n}$, whenever, for all valuations $\fdec{\nu}{\enset{x_1, \cdots, x_n}}{B \subseteq A}$, we have $\rmean{p}_{\nu} =  \rmean{q}_{\nu}$.
$B$ is said to \emph{satisfy a system $E$ of equations} if it satisfies all $\pair{p,q} \in E$. Further, we define 
\begin{equation}
			Sub_{\aA}/E\; =\; \setdef{B \in Sub_{\aA}}{B\; \text{satisfies}\; E}.
\end{equation}
\end{defn}

We will now consider the poset $\pair{Sub_{\aA}/E, \subseteq}$, which is a subposet of $\pair{Sub_{\aA}, \subseteq}$ considered before. Clearly, if $\aA$ satisfies $E$, we have $Sub_{\aA}/E = Sub_{\aA}$, which is not very interesting. In all other cases $Sub_{\aA}/E$ has no top element. In fact, it is not even a lattice in most cases. However, we still have some weakened form of completeness:

\begin{prop}\label{SubEdomain}
$Sub_{\aA}/E$ is a bounded-complete algebraic domain, i.e. it is a bc-dcpo and it is algebraic.
\end{prop}
\begin{proof}
  To proof that it is an algebraic domain, it is enough to show that $Sub_{\aA}/E$ is a Scott-closed subset of $Sub_\aA$,
  i.e., that it is closed for directed suprema (hence a dcpo) and downwards closed (hence, given the first condition, $\ddown x$ is the same in both posets and algebraicity follows from the same property on $Sub_\aA$).
  Downwards closeness follows immediately, since a subset of an set satisfying an equation also satisfies it. 

  For directed completeness, 
  let $\aS$ be a directed subset of $Sub_{\aA}/E$.
  We want to show that its supremum, the union $\bigcup \aS$, is still in $Sub_{\aA}/E$.
  Let $\pair{p,q} \in E$
be an equation over the variables $\enset{x_1, \ldots, x_n}$ and $\fdec{\nu}{\enset{x_1, \ldots, x_n}}{\bigcup \aS}$ any valuation on $\bigcup \aS$. Let us write $a_1, \ldots, a_n$ for $\nu(x_1), \ldots, \nu(x_n)$. There exist $S_1, \ldots, S_n \in \aS$ such that $a_1 \in S_1, \ldots, a_n \in S_n$. By directedness of $\aS$, there is an $S \in \aS$ such that $\enset{a_1, \ldots, a_n} \subseteq S$. Since $S \in Sub_{\aA}/E$ and the valuation $\nu$ is defined on $S$, one must have $\rmean{p}_{\nu} =  \rmean{q}_{\nu}$. Therefore, $\bigcup \aS$ satisfies the equations in $E$ and
so $Sub_{\aA}/E$ is closed under directed suprema.

For bounded completeness, note that if $\aS \subseteq Sub_{\aA}/E$ is bounded above by a subalgebra $S'$ that also satisfies the equations, then the subalgebra generated by $\aS$, $\left\langle \bigcup\aS \right\rangle$, being a subset of $S'$, must also satisfy the equations. So $\left\langle \bigcup\aS \right\rangle$ is in $Sub_{\aA}/E$ and is a supremum for $\aS$ in this poset.
\end{proof}

\medskip
\paragraph{Topologically closed subalgebras.} The results of the previous section apply to any kind of algebraic structure. In particular, this is enough to conclude that, given a $*$-algebra $\N$, the set of its abelian $*$-subalgebras forms an algebraic domain. If we restrict attention to the finite-dimensional situation, \ie matrix algebras $\N$, or finite direct sums of matrix algebras, then the result shows that the poset $\VN$ of abelian von Neumann subalgebras of a matrix algebra $\N$ is an algebraic domain, since every algebraically closed abelian $*$-subalgebra is also weakly closed (and hence a von Neumann algebra) in this case.

But for a general von Neumann algebra $\N$, not all abelian $*$-subalgebras need to be abelian von Neumann subalgebras. We need to consider the extra condition that each given subalgebra is closed with respect to a certain topology, namely the weak operator
topology (or the strong operator topology, or the $\sigma$-weak topology for that matter).

Again, we follow a general path, proving what assertions can be made about posets of subalgebras of any kind of algebraic structures, where
its subalgebras are also topologically closed.

For the rest of this section, we only consider Hausdorff topological spaces. This condition is necessary for our proofs to work. The reason is that, in a Hausdorff space, a net $(a_i)_{i \in I}$ converges to at most one point $a$. We shall write this as $(a_i)_{i \in I} \longrightarrow a$.

\begin{defn}
A \emph{topological algebra} is an algebra $\aA = \pair{A;\aF}$ where $A$ is equipped with a Hausdorff topology.
\end{defn}
We single out the substructures of interest:
\begin{defn}
Given a \emph{topological algebra} $\aA = \pair{A;\aF}$, a \emph{closed subalgebra} is
a subset $B$ of $A$ which is simultaneously a subalgebra and a topologically closed set. The set of closed subalgebras of $\aA$ is denoted by $CSub_{\aA}$. Moreover, for a system of polynomial equations $E$ over
 $\aA$, $CSub_{\aA}/E$ denotes $CSub_{\aA} \cap Sub_{\aA}/E$.
 \end{defn}
 
Our goal is to extend the results about $Sub_{\aA}/E$ to $CSub_{\aA}/E$. For this, we need to impose a topological condition
on the behaviour of the operations of $\aA$.
Usually, one requires the algebraic operations to be continuous, which would allow us to prove results regarding completeness which are similar to those in previous subsections. An example would then be the poset of (abelian) closed subgroups of a topological groups. However, multiplication in a von Neumann algebra is not continuous with respect to the weak operator topology, only separately continuous in each argument. Thus, in order to capture the case of $\VN$, we need to weaken the continuity assumption. The following condition will suffice:
 
\begin{defn}
Let $\aA = \pair{A;\aF}$ be a topological algebra and $f \in \aF$ an operation with arity $n$. 
We say that $f$ is \emph{separately continuous} if, for any $k = 1, \ldots, n$ and elements $b_1, \ldots, b_{k-1}, b_{k+1}, \ldots, b_n \in A$,
the function
\begin{align*}
\aA\; & \longrightarrow\; \aA\\
a\; & \longmapsto\; f(b_1, \cdots, b_{k-1}, a, b_{k+1}, \cdots, b_n) 
\end{align*}
is continuous. Equivalently, one can say that for any net $(a_i)_{i \in I}$ such that
$(a_i)_{i \in I} \longrightarrow a$, we have 
\[
(f(b_1, \cdots, b_{k-1}, a_i, b_{k+1}, \cdots, b_n))_{i \in I} \longrightarrow f(b_1, \cdots, b_{k-1}, a, b_{k+1}, \cdots, b_n).
\]
\end{defn}

The fact that we allow a weaker form of continuity than it is costumary on the algebraic operations forces us to
impose a condition on the allowed equations.
\begin{defn}
A polynomial over $\aA$ in the variables $x_1, \cdots, x_n$ is \emph{linear} if each of the variables occurs at most once.
\end{defn}

\begin{lem}\label{polycts}
Let $\aA = \pair{A;\aF}$ be a topological algebra with separately continuous operations,
and $p$ a linear polynomial over $\aA$ in the variables $x_1, \cdots, x_n$. Then the function
\begin{align*}
\widetilde{p} \; :\; & \aA^n \longrightarrow \aA\\
(a_1, \cdots, a_n) & \mapsto \rmean{p}_{\fdef{\nu}{x_{i}}{a_{i}}}
\end{align*}
is separately continuous. 
\end{lem}
\begin{proof}
The proof goes by induction on (linear) polynomials (refer back to definitions \ref{def:poly} and \ref{def:evaluatepoly}):
\begin{itemize}
\item If $p = x_j$, then 

\begin{equation}
			\widetilde{p}(a_1, \cdots, a_n)\; =\; \rmean{x_j}_{\fdec{\nu}{x_{i}}{a_{i}}} \; =\; \nu(x_j) \; =\; a_j.
\end{equation}
So $\widetilde{p} = \pi_j$, which is clearly separately continuous (either constant or identity when arguments taken separately).
\item If $p = f(p_1, \cdots, p_l)$, assume as the induction hypothesis that  $\widetilde{p_1},  \ldots, \widetilde{p_l}$
are separately continuous. Then
\begin{equation}
			\widetilde{p} \; =\; f \comp \pair{\widetilde{p_1},  \ldots, \widetilde{p_l}}.
\end{equation}
Let us see that this is separately continuous in the first argument $x_1$ (the other arguments can be treated analogously). Let $b_2 \cdots b_n \in A$. Since $p$ is linear, the variable $x_1$ occurs on at most one subpolynomial $p_k$  (with $k=1, \ldots, l$). Without  loss of generality,
say it occurs on $p_1$. Then the functions 
\[
t_k \; =\; a \longmapsto \widetilde{p_k} (a, b_2, \ldots, b_n)
\]
for $k = 2, \cdots, l$ are constant (since $x_1$ does on occur in $p_k$), whereas the function
\[
t_1 \; =\; a \longmapsto \widetilde{p_1} (a, b_2, \ldots, b_n)
\]
is continuous (by the induction hypothesis).

Now, we consider the function $\widetilde{p}$ with only the first argument varying. This is the function $f \comp \pair{t_1, \cdots, t_l}$. But $t_2, \cdots, t_l$ are constant functions, which means that only the first argument of $f$ varies. Since this is given by a continuous function $t_1$ and $f$ is continuous in the first argument (because it is separately continuous), $\fdec{f \comp  \pair{t_1, \cdots, t_l}}{A}{A}$ is continuous, meaning that $\widetilde{p}$ is (separately) continuous in the first argument. For the other arguments, the proof is similar.
\end{itemize}
\end{proof}

We shall denote by $cl(-)$ the (Kuratowski) closure operator associated with the topology of $\aA$.

\begin{lem}\label{closuresubalg}
For a topological algebra $\aA = \pair{A;\aF}$ with separately continuous operations,
\begin{enumerate}
	\item\label{clsub-1} If $B \in Sub_{\aA}$, then $cl(B) \in Sub_{\aA}$.
	\item\label{clsub-2} Given a system of linear polynomial equations $E$ over $\aA$, if $B \in Sub_{\aA}/E$, then $cl(B) \in Sub_{\aA}/E$.
\end{enumerate}
\end{lem}
\begin{proof}
(\ref{clsub-1}) We need to show that $cl(B)$ is closed under the operations in $\aF$. We consider only the case of a binary operation $f \in \aF$. It will be apparent that the general case follows from a simple inductive argument.

Let $a,b \in cl(B)$. Then there exist nets $(a_i)_{i \in I}$ and $(b_j)_{j \in J}$ consisting of elements of $B$ such that $(a_i)_{i \in I} \longrightarrow a$ and $(b_j)_{j \in J}  \longrightarrow b$. By fixing an index $j$ at a time, we conclude, by separate continuity, that
\begin{equation}
			\forall{j \in J}: (f(a_i, b_j))_{i \in I} \longrightarrow f(a,b_j).
\end{equation}
Since all the elements of the net $(f(a_i, b_j))_{i \in I}$ are in $B$ (for $a_i, b_j \in B$ and $B$ is closed under $f$), all elements of the form $f(a,b_j)$ are in $cl(B)$, as they are limits of nets of elements of $B$. Now, letting $j$ vary again, by separate continuity we obtain
\begin{equation}
			(f(a, b_j))_{j \in J} \longrightarrow f(a,b).
\end{equation}
But $(f(a, b_j))_{j \in J}$ is a net of elements of $cl(B)$, thus $f(a,b) \in cl(cl(B)) = cl(B)$.
This completes the proof that $cl(B)$ is closed under $f$.

(\ref{clsub-2}) The procedure is very similar and again we consider only a polynomial equation $\pair{p,q}$ in two variables $x$ and $y$. We will prove that $cl(B)$ satisfies the equation whenever $B$ does.

Let $a,b \in cl(B)$, with sets $(a_i)_{i \in I}$ and $(b_j)_{j \in J}$ as before.
We use the function $\fdec{\widetilde{p}}{A^2}{A}$ from lemma \ref{polycts}, which is given by 
\begin{equation}
			\widetilde{p}(k_1, k_2) = \rmean{p}_{[x_1\mapsto k_1, x_2\mapsto k_2]}.
\end{equation}
Because of Lemma \ref{polycts}, we can follow the same procedure as in the first part (for the separately continuous function $f$) to conclude
that
\begin{equation}
  \forall{j \in J}: (\widetilde{p}(a_i, b_j))_{i \in I} \longrightarrow \widetilde{p}(a, b_j) \;\;\; \wedge \;\;\;  (\widetilde{q}(a_i, b_j))_{i \in I} \longrightarrow \widetilde{q}(a, b_j).
\end{equation}
But $a_i, b_j \in B$ and $B$ satisfies the equations, so the nets $(\widetilde{p}(a_i, b_j))_{i \in I}$ and $(\widetilde{q}(a_i, b_j))_{i \in I}$ are the same. 
Since we have Hausdorff spaces by assumption, a net has at most one limit. This implies that
\begin{equation}\label{eq:Step1Conc}
			\forall{j \in J}: \widetilde{p}(a, b_j) = \widetilde{q}(a, b_j).
\end{equation}
As before, we take (again by separate continuity of $\widetilde{p}$ and $\widetilde{q}$)
\begin{equation}\label{eq:Step2}
			(\widetilde{p}(a, b_j))_{j \in J} \longrightarrow \widetilde{p}(a, b)
			\;\;\; \wedge \;\;\;
			(\widetilde{q}(a, b_j))_{j \in J} \longrightarrow \widetilde{q}(a, b). 
\end{equation}
Because of \ref{eq:Step1Conc} and \ref{eq:Step2}, the nets are the same again, yielding
\begin{equation}
			\widetilde{p}(a, b) \; =\;   \widetilde{q}(a, b).
\end{equation}
This proves that $cl(B)$ satifies $\pair{p,q}$.
\end{proof}


We now can state our result.
As far as completeness (or convergence) properties are concerned, we have a similar situation as in the non-topological case.
\begin{prop}			\label{Prop_CSub_ACompleteLat}
Let $\aA = \pair{A;\aF}$ be a topological algebra with separately continuous operations. Let $E$ be a system of linear polynomial equations.
Then
\begin{itemize}
\item $CSub_{\aA}$ is a complete lattice.
\item $CSub_{\aA}/E$ is a bounded-complete dcpo.
\end{itemize}
\end{prop}
\begin{proof}
  If $S$ is a set of closed subalgebras, then, by \ref{closuresubalg}-\ref{clsub-1}, there is a least closed subalgebra containing $S$, given by $cl(\pair{\bigcup S})$. This proves the first claim.

Now, suppose that all $B \in S$ satisfy $E$ and that $S$ is directed. By directedness, the supremum is given by 
$cl(\pair{\bigcup S}) = cl(\bigcup S)$. By Lemma \ref{SubEdomain}, $\bigcup S$ satisfies the equations $E$. Then, by \ref{closuresubalg}-\ref{clsub-2}, $cl(\bigcup S)$ also does.

Similarly, suppose all $B\in S$ satisfy $E$ and that $S$ is bounded above by $C \in CSub_{\aA}/E$. Then the supremum of $S$ in $CSub_{\aA}$, namely $cl(\pair{\bigcup S})$, is a closed subalgebra smaller than $C$. Therefore, $cl(\pair{\bigcup S})$ must also satisfy $E$, hence it is a supremum for $S$ in $CSub_{\aA}/E$.
\end{proof}

\medskip
\paragraph{$\VN$ as a dcpo.} The results above apply easily to the situation we are interested in, namely the poset of abelian von Neumann subalgebras of a nonabelian von Neumann algebra $\N$. We can describe $\N$ as the algebra\footnote{Algebra here has two different meanings. The first is the (traditional) meaning, of which the $*$-algebras, $C^{*}$-algebras and von Neumann are special cases. When we say just algebra, though, we mean it in the sense of universal algebra as in the previous two sections.}
\begin{equation}
  \pair{\N; \one, +, \comp_{c}, \times, *}
\end{equation}
whose operations have arities $0$, $2$, $1$, $2$ and $1$, respectively.  The operations are just the usual $*$-algebra operations. Note that $ \comp_{c}$ (scalar multiplication by  $c \in \Cc$) is in fact an uncountable set of operations indexed by $\Cc$. Note that the inclusion of $\one$ as an operation allows to restrict our attention to unital subalgebras with the same unit and to unital homomorphisms. Recall that a $C^*$-subalgebra of $\N$ is a $*$-subalgebra that is closed with respect to the norm topology. A von Neumann subalgebra is a $*$-subalgebra that is closed with respect to the weak operator topology. Since both these topologies are Hausdorff and the operations are separately continuous with respect to both, the results from the previous sections apply.

\begin{prop}
Let $\N$ be a von Neumann algebra. Then
\begin{itemize}
	\item the set of unital $*$-subalgebras (respectively, $C^{*}$-subalgebras, and von Neumann subalgebras) of $\N$ with the same unit element as $\N$ is a complete lattice.
\item The set of abelian unital $*$-subalgebras (respectively, $C^{*}$-subalgebras, and von Neumann subalgebras) of $\N$ with the same unit element as $\N$ is a bounded complete dcpo (or complete semillatice).
\end{itemize}
\end{prop}
The set of abelian von Neumann subalgebras of $\N$ with the same unit element as $\N$ is simply $\VN$. Hence, we have shown that $\VN$ is a dcpo. Moreover, it is clear that $\VN$ is bounded complete if we include the trivial algebra as the bottom element. Note that this holds for \emph{all} types of von Neumann algebras. Furthermore, our proof shows that analogous results hold for the poset of abelian $C*$-subalgebras of unital $C^*$-algebras (and other kinds of topological algebras and equations, such as associative closed Jordan subalgebras, or closed subgroups of a topological group).

As already remarked at the beginning of this subsection on topologically closed subalgebras, $\VN$ is an algebraic domain in the case that $\N$ is a finite-dimensional algebra (finite direct sum of matrix algebras over $\Cc$). This can be seen directly from the fact that, in the finite-dimensional situation, any subalgebra is weakly (resp. norm) closed. Hence, the topological aspects do not play a r\^ole and proposition \ref{SubEdomain} applies. The result that $\VN$ is an algebraic domain if $\N$ is a finite-dimensional matrix algebra can also be deduced from the fact that $\VN$ has finite height in this case. This clearly means that any directed set has a maximal element and, consequently, the way-below relation coincides with the order relation, thus $\VN$ is trivially an algebraic domain.

One can also see directly that the context category $\VN$ is a dcpo for an arbitrary von Neumann algebra $\N$: let $S\subset\VN$ be a directed subset, and let
\begin{equation}
			A:=\pair{\bigcup_{V\in S} V}
\end{equation}
be the algebra generated by the elements of $S$. Clearly, $A$ is a self-adjoint abelian algebra that contains $\hat 1$, but not a von Neumann algebra in general, since it need not be weakly closed. By von Neumann's double commutant theorem, the weak closure $\tilde V:=\overline{A}^w$ of $A$ is given by the double commutant $A''$ of $A$. The double commutant construction preserves commutativity, so $A''=\tilde V$ is an \emph{abelian} von Neumann subalgebra of $\N$, namely the smallest von Neumann algebra containing all the $V\in S$, so
\begin{equation}
			\tilde V=\dsup S.
\end{equation}
Hence, every directed subset $S$ has a supremum in $\VN$.
\medskip

\paragraph{A functor from von Neumann algebras to dcpos.} We have seen that to each von Neumann algebra $\N$, we can associate a dcpo $\VN$. We will now show that this assignment is functorial. Again, this result arises as a special case of a more general proposition,
which is quite easy to show.

\begin{defn}
  Let $\mathcal{A}$ and $\mathcal{B}$ be two algebras with the same signature $\Sigma$ (that is, they have the same set of operations).
  A function $\fdec{\phi}{A}{B}$ is called a \emph{homomorphism} if it preserves every operation. That is,
for each operation $f$ of arity $n$,
\begin{equation}
			\phi(f^{\mathcal{A}}(a_1,\ldots,a_n)) = f^{\mathcal{B}}(\phi(a_1),\ldots,\phi(a_n)).
\end{equation}
\end{defn}

\begin{lem}\label{MorphismsScott}
Let $\mathcal{A}$ and $\mathcal{B}$ be topological algebras and $\fdec{\phi}{A}{B}$ a continuous homomorphism that preserves closed subalgebras. Then, for any $E$, the function
\begin{align}
    \phi^\rightarrow|_{CSub_\mathcal{A}/E} : CSub_\mathcal{A}/E &\longrightarrow CSub_\mathcal{B}/E\\			\nonumber
    S &\longmapsto \phi^\rightarrow{S}
\end{align}
is a Scott-continuous function.
\end{lem}
\begin{proof}
  First, we check that the function is well-defined, that is, each element of the domain determines one element in the codomain. Let $S \in  CSub_\mathcal{A}/E$. Then $S':=\phi^\rightarrow(S)$ is in $CSub_\mathcal{B}$ by the assumption that $\phi$ preserves closed subalgebras. The homomorphism condition implies that it also preserves the satisfaction of equations $E$.

  As for the result itself, note that monotonicity is trivial. So we only need to show that this map preserves directed joins. We know that the directed join is given as $\dsup D = cl(\bigcup D)$. Thus, we want to prove that for any directed set $D$
\begin{equation*}
  \phi^\rightarrow(cl(\bigcup_{S \in D}S))
  \subseteq
  cl(\bigcup_{S \in D}\phi^\rightarrow(S))
\end{equation*}
This follows easily since $f^\rightarrow$ commutes with union of sets and the inequality
$\phi^\rightarrow{cl(S)} \subseteq cl(\phi^\rightarrow{S})$
is simply the continuity of $\phi$.
\end{proof}

We can combine the previous lemma with Prop. \ref{Prop_CSub_ACompleteLat} as follows:
\begin{prop}
Let $\mathcal{C}$ be a category such that
\begin{itemize}
  \item objects of $\mc C$ are topological $\Sigma$-algebras (for a fixed signature $\Sigma$) that fulfil the conditions stated in Prop. \ref{Prop_CSub_ACompleteLat},
  \item morphisms of $\mc C$ are continuous homomorphisms preserving closed subalgebras.
\end{itemize}
Then, for any system of linear polynomial equations $E$, the assignments 
\begin{align}
			A &\longmapsto CSub_A/E\\
			(\fdec{\phi}{A}{B}) &\longmapsto (\fdec{\phi^\rightarrow|_{CSub_A/E}}{CSub_A/E}{CSub_B/E})
\end{align}
define a functor $\fdec{F_E}{\mathcal{C}}{\sf{DCPO}}$ from $\mathcal{C}$ to the category of dcpos and Scott-continuous functions.
\end{prop}
\begin{proof}
  \ref{Prop_CSub_ACompleteLat} and \ref{MorphismsScott} imply that this is a well-defined between the categories. Functoriality is immediate from the properties of images $f^\rightarrow$ of a function $f$: 
  $id_A^\rightarrow = id_{\mathcal{P} A}$ and $f^\rightarrow \circ g^\rightarrow = (f \circ g)^\rightarrow$.
\end{proof}

This result can be applied to our situation. Let $\N$ be a von Neumann algebra. Recall that a von Neumann subalgebra is a $*$-subalgebra closed with respect to the $\sigma$-weak topology.\footnote{Equivalently, a von Neumann algebra is closed in the weak operator topology, as mentioned before.} This topology is Hausdorff, and $\sigma$-weakly continuous (or normal) $*$-homomorphisms map von Neumann subalgebras to von Neumann subalgebras. Hence, we get
\begin{thm}
The assignments
\begin{align}
  	\nonumber	\mc V:\sf{vNAlg} &\lra \sf{DCPO}\\
			\N &\lmt \VN\\
			\phi &\lmt \phi^\ra|_{\VN}
\end{align}
define a functor from the category $\sf{vNAlg}$ of von Neumann algebras and $\sigma$-weakly continuous, unital algebra homomorphisms to the category $\sf{DCPO}$ of dcpos and Scott-continuous functions.
\end{thm}

It is easy to see that the functor $\mc V$ is not full. It is currently an open question whether $\mc V$ is faithful.

\section{The context category $\VN$ in infinite dimensions -- lack of continuity}			\label{Sec_VNNoDomainInInfDim}
\paragraph{The case of $\N=\BH$ with infinite-dimensional $\Hi$.} We now show that for an infinite-dimensional Hilbert space $\Hi$ and the von Neumann algebra $\N=\BH$, the poset $\VN$ is \emph{not} continuous, and hence not a domain. This counterexample is due to Nadish de Silva, whom we thank.

Let $\Hi$ be a separable, infinite-dimensional Hilbert space, and let $(e_i)_{i\in\bbN}$ be an orthonormal basis. For each $i\in\bbN$, let $\hP_i$ be the projection onto the ray $\bbC e_i$. This is a countable set of pairwise-orthogonal projections.

Let $\hP_e:=\sum_{i=1}^\infty\hP_{2i}$, that is, $\hP_e$ is the sum of all the $\hP_i$ with even indices. Let $V_e:=\{\hP_e,\hat 1\}''=\bbC\hP_e+\bbC\hat 1$ be the abelian von Neumann algebra generated by $\hP_e$ and $\hat 1$. Moreover, let $V_m$ be the (maximal) abelian von Neumann subalgebra generated by all the $\hP_i$, $i\in\bbN$. Clearly, we have $V_e\subset V_m$.

Now let $V_j:=\{\hP_1,...,\hP_j,\hat 1\}''=\bbC\hP_1+...+\bbC\hP_j+\bbC\hat 1$ be the abelian von Neumann algebra generated by $\hat 1$ and the first $j$ projections in the orthonormal basis, where $j\in\bbN$. The algebras $V_j$ form a directed set (actually a chain) $(V_j)_{j\in\bbN}$ whose directed join is $V_m$, obviously. But note that none of the algebras $V_j$ contains $V_e$, since the projection $\hP_e$ is not contained in any of the $V_j$. Hence, we have the situation that $V_e$ is contained in $\dsup (V_j)_{j\in\bbN} = V_m$, but $V_e\nsubseteq V_j$ for any element $V_j$ of the directed set, so $V_e$ is not way below itself.

Since $V_e$ is an atom in the poset $\VN$ (i.e. only the bottom element, the trivial algebra, is below $V_e$ in the poset), we have that $\ddown V_e=\{\bot\} = \{\bbC\hat 1\}$, so $\dsup\ddown V_e = \dsup\{\bot\} = \bot \neq V_e$, which implies that $\VN$ is not continuous.
\medskip
\paragraph{Other types of von Neumann algebras}
The proof that $\VN=\mc V(\BH)$ is not a continuous poset if $\Hi$ is infinite-dimensional can be generalised to other types of von Neumann algebras. ($\BH$ is a type $I_\infty$ factor.) We will use the well-known fact that an abelian von Neumann algebra $V$ is of the form
\begin{equation}
			V\simeq \ell^\infty(X,\mu)
\end{equation}
for some measure space $(X,\mu)$. It is known that the following cases can occur: $X=\{1,...,n\}$, for each $n\in\bbN$, or $X=\bbN$, equipped with the counting measure $\mu_c$; or $X=[0,1]$, equipped with the Lebesgue measure $\mu$; or the combinations $X=\{1,...,n\}\sqcup [0,1]$ and $X=\bbN\sqcup [0,1]$. 

A maximal abelian subalgebra $V_m$ of $\BH$ generated by the projections onto the basis vectors of an orthonormal basis $(e_i)_{i\in\bbN}$ of an infinite-dimensional, separable Hilbert space $\Hi$ corresponds to the case of $X=\bbN$, so $V_m\simeq\ell^\infty(\bbN)$. The proof that in this case $\mc V(\BH)$ is not continuous is based on the fact that $X=\bbN$ can be decomposed into a countable family $(S_i)_{i\in\bbN}$ of measurable, pairwise disjoint subsets with $\bigcup_i S_i=X$ such that $\mu(S_i)>0$ for each $i$. Of course, the sets $S_i$ can just be taken to be the singletons $\{i\}$, for $i\in\bbN$, in this case. These are measurable and have measure greater zero since the counting measure on $X=\bbN$ is used.

Yet, also for $X=([0,1],\mu)$ we can find a countable family $(S_i)_{i\in\bbN}$ of measurable, pairwise disjoint subsets with $\bigcup_i S_i=X$ such that $\mu(S_i)>0$ for each $i$: for example, take $S_1:=[0,\frac{1}{2})$, $S_2:=[\frac{1}{2},\frac{3}{4})$, $S_3:=[\frac{3}{4},\frac{7}{8})$, etc. Each measurable subset $S_i$ corresponds to a projection $\hP_i$ in the abelian von Neumann algebra $V\simeq\ell^\infty([0,1],\mu)$, so we obtain a countable family $(\hP_i)_{i\in\bbN}$ of pairwise orthogonal projections with sum $\hat 1$ in $V$. Let $\hP_e:=\sum_{i=1}^\infty \hP_{2i}$, and let $V_e:=\{\hP_e,\hat 1\}''=\bbC\hP_e+\bbC\hat 1$.

Moreover, let $V_j:=\{\hP_1,...,\hP_j,\hat 1\}''$ be the abelian von Neumann algebra generated by $\hat 1$ and the first $j$ projections, where $j\in\bbN$. Clearly, $(V_j)_{j\in\bbN}$ is a directed set, and none of the algebras $V_j$, $j\in\bbN$, contains the algebra $V_e$, since $\hP_e\notin V_j$ for all $j$. Let $V:=\dsup V_j$ be the abelian von Neumann algebra generated by all the $V_j$. Then we have $V_e\subset V$, so $V_e$ is not way below itself and $\dsup\ddown V_e\neq V_e$.

This implies that $\VN$ is not continuous if $\N$ contains an abelian subalgebra of the form $V\simeq\ell^\infty([0,1],\mu)$. Clearly, analogous results hold if $\N$ contains any abelian subalgebra of the form $V=\ell^\infty(\{1,...,n\}\sqcup[0,1])$ or $V\simeq\ell^\infty(\bbN\sqcup[0,1])$ -- in none of these cases is $\VN$ a domain.

A von Neumann algebra $\N$ contains only abelian subalgebras of the form $V\simeq\ell^\infty(\{1,...,n\},\mu_c)$ if and only if $\N$ is either (a) a finite-dimensional matrix algebra $M_n(\bbC)\otimes\mc C$ with $\mc C$ an abelian von Neumann algebra (i.e., $\N$ is a type $I_n$-algebra) such that the center $\mc C$ does not contain countably many pairwise orthogonal projections,\footnote{This of course means that $\mc C\simeq\ell^\infty(\{1,...,k\},\mu_c)$, that is, $\mc C$ is isomorphic to a finite-dimensional algebra of diagonal matrices with complex entries.} or (b) a finite direct sum of such matrix algebras. Equivalently, $\N$ is a \emph{finite-dimensional von Neumann algebra}, that is, $\N$ is represented (faithfully) on a finite-dimensional Hilbert space. We have shown:

\begin{thm}
The context category $\VN$, that is, the partially ordered set of abelian von Neumann subalgebras of a von Neumann algebra $\N$ which share the unit element with $\N$, is a continuous poset and hence a domain if and only if $\N$ is finite-dimensional.
\end{thm}
\medskip
\section{Outlook: Space and Space-time from Interval Domains?}			\label{Sec_Outlook}
There are many further aspects of the link between the topos approach and domain theory to be explored. In this brief section, we indicate some very early ideas about an application to concepts of space and space-time.\footnote{The reader not interested in speculation may well skip this section!} The ideas partly go back to discussions with Chris Isham, but we take full responsibility.

Let us take seriously the idea suggested by the topos approach that in quantum theory it is not the real numbers $\bbR$ where physical quantities take their values, but rather the presheaf $\Rlr$. This applies in particular to the physical quantity position. We know from ordinary quantum theory that depending on its state, a quantum particle will be localised more or less well. (Strictly speaking, there are no eigenstates of position, so a particle can never be localised perfectly.) The idea is that we can now describe `unsharp positions' as well as sharp ones by using the presheaf $\Rlr$ and its global elements. A quantum particle will not `see' a point in the continuum $\bbR$ as its position, but, depending on its state, it will `see' some global element of $\Rlr$ as its (generalised) position.

These heuristic arguments suggest the possibility of building a model of space based not on the usual continuum $\bbR$, but on the quantity-value presheaf $\Rlr$. The real numbers are embedded (topologically) in $\Rlr$ as a tiny part, they correspond to global elements $\ga=(\mu,\nu)$ such that $\mu=\nu=const_a$, where $a\in\bbR$ and $const_a$ is the constant function with value $a$.

It is no problem to form $\Rlr^3$ or $\Rlr^n$ within the topos, and also externally it is clear what these presheaves look like. It is also possible to generalise notions of distance and metric from $\bbR^n$ to $\IR^n$, \ie powers of the interval domain. This is an intermediate step to generalising metrics from $\bbR^n$ to $\Rlr^n$. Yet, there are many open questions.

The main difficulty is to find a suitable, physically sensible group that would act on $\Rlr^n$. Let us focus on space, modelled on $\Rlr^3$, for now. We picture an element of $\Rlr^3_V$ as a `box': there are three intervals in three independent coordinate directions. We see that in this interpretation, there implicitly is a spatial coordinate system given. Translations are unproblematic, they will map a box to another box. But rotations will map a box to a rotated box that in general cannot be interpreted as a box with respect to the same (spatial) coordinate system. This suggests that one may want to consider a \emph{space of (box) shapes} rather than the boxes themselves. These shapes would then represent generalised, or `unsharp', positions. Things get even more difficult if we want to describe space-time. There is no time operator, so we have to take some liberty when interpreting time intervals as `unsharp moments in time'. More crucially, the Poincar\'e group acting on classical special relativistic space-time is much more complicated than the Galilei group that acts separately on space and time, so it will be harder to find a suitable `space of shapes' on which the Poincar\'e group, or some generalisation of it, will act.

There is ongoing work by the authors to build a model of space and space-time based on the quantity-value presheaf $\Rlr$. It should be mentioned that classical, globally hyperbolic space-times can be treated as domains, as was shown by Martin and Panangaden \cite{MP06}. These authors consider generalised interval domains over so-called bicontinuous posets. This very interesting work serves as a motivation for our ambitious goal of defining a notion of generalised space and space-time suitable for quantum theory and theories `beyond quantum theory', in the direction of quantum gravity and quantum cosmology. Such a model of space and space-time would not come from a discretisation, but from an embedding of the usual continuum into a much richer structure. It is very plausible that domain-theoretic techniques will still apply, which could give a systematic way to define a classical limit of these quantum space-times. For now, we merely observe that the topos approach, which so far applies to non-relativistic quantum theory, provides suggestive structures that may lead to new models of space and space-time.

\section{Conclusion}			\label{Sec_Conclusion}
We developed some connections between the topos approach to quantum theory and domain theory. The topos approach provides new `spaces' of physical relevance, given as objects in a presheaf topos. Since these new spaces are not even sets, many of the usual mathematical techniques do not apply, and we have to find suitable tools to analyse them. In this article, we focused on the quantity-value object $\Rlr$ in which physical quantities take their values, and showed that domain theory provides tools for understanding the structure of this space of generalised, unsharp values. This also led to some preliminary ideas about models of space and space-time that would incorporate unsharp positions, extending the usual continuum picture of space-time.
\medskip

\subsection*{Acknowledgments}
We thank Chris Isham, Prakash Panangaden and Samson Abramsky for valuable discussions and encouragement. Many thanks to Nadish de Silva for providing the counterexample in Section \ref{Sec_VNNoDomainInInfDim}, and to the anonymous referee for a number of valuable suggestions and hints. RSB gratefully acknowledges support by the Marie Curie Initial Training Network in Mathematical Logic  - MALOA - From MAthematical LOgic to Applications, PITN-GA-2009-238381, as well as previous support by Santander Abbey bank. We also thank Felix Finster, Olaf M\"uller, Marc Nardmann, J\"urgen Tolksdorf and Eberhard Zeidler for organising the very enjoyable and interesting ``Quantum Field Theory and Gravity'' conference in Regensburg in September 2010.


\begin{thebibliography}{1}

\bibitem{AJ94} S.~Abramsky, A.~Jung.
\newblock Domain Theory.
\newblock In \textit{Handbook of Logic in Computer Science}, eds. S.~Abramsky, D.M.~Gabbay, T.S.E.~Maibaum, Clarendon Press, 1--168 (1994).

\bibitem{CHLS09}{M.~Caspers, C.~Heunen, N.P.~Landsman, B.~Spitters.}
\newblock Intuitionistic quantum logic of an {$n$}-level system.
\newblock \emph{Found.\ Phys.} \textbf{39}, 731--759 (2009).

\bibitem{CM11} B.~Coecke, K.~Martin.
\newblock A Partial Order on Classical and Quantum States.
\newblock In \textit{New Structures for Physics}, ed. B.~Coecke, Springer, Heidelberg (2011).

\bibitem {Doe05} A.~D\"{o}ring.
\newblock Kochen-Specker theorem for von Neumann algebras.
\newblock {\em Int.\ Jour.\ Theor.\ Phys.} {\bf 44}, 139--160 (2005).

\bibitem{Doe07b} A.~D\"{o}ring.
\newblock Topos theory and `neo-realist' quantum theory.
\newblock In \textit{Quantum Field Theory, Competitive Models}, eds. B. Fauser, J. Tolksdorf, E. Zeidler, Birkh\"auser (2009).

\bibitem{Doe08} A.~D\"{oring}.
\newblock Quantum States and Measures on the Spectral Presheaf.
\newblock \emph{Adv. Sci. Lett.} \textbf{2}, special issue on ``Quantum Gravity, Cosmology and Black Holes'', ed. M. Bojowald, 291--301 (2009).

\bibitem{Doe09} A.~D\"{o}ring.
\newblock Topos quantum logic and mixed states.
\newblock In \textit{Proceedings of the 6th International Workshop on Quantum Physics and Logic (QPL 2009), Oxford}, eds. B. Coecke, P. Panangaden, and P. Selinger.
\newblock \textit{Electronic Notes in Theoretical Computer Science} \textbf{270}(2) (2011).

\bibitem{Doe10} A.~D\"oring.
\newblock The Physical Interpretation of Daseinisation.
\newblock In \textit{Deep Beauty}, ed. Hans Halvorson, 207--238, Cambridge University Press, New York (2011).

\bibitem {DI(1)} A.~D\"{o}ring, and C.J. Isham.
\newblock A topos foundation for theories of physics:
    {I.} Formal languages for physics.
\newblock \emph{J. Math. Phys} {\bf 49}, 053515 (2008).

\bibitem {DI(2)} A.~D\"{o}ring, and C.J. Isham.
\newblock A topos foundation for theories of physics:
    {II.} Daseinisation and the liberation of quantum theory.
\newblock \emph{J. Math. Phys} {\bf 49}, 053516 (2008).

\bibitem {DI(3)} A.~D\"{o}ring, and C.J. Isham.
\newblock A topos foundation for theories of physics:
    {III.} Quantum theory and the representation of
    physical quantities with arrows \mbox{$\breve{A}:\Sig\map\ps\R$}.
\newblock  \emph{J. Math. Phys} {\bf 49}, 053517 (2008).

\bibitem{DI(4)} A.~D\"{o}ring, and C.J. Isham.
\newblock A topos foundation for theories of physics:
            {IV.} Categories of systems.
\newblock \emph{J. Math. Phys} {\bf 49}, 053518 (2008).

\bibitem{DI(Coecke)08}  A.~D\"{o}ring, and C.J. Isham.
\newblock `What is a Thing?': Topos Theory in the Foundations of Physics.
\newblock In \emph{New Structures for Physics}, ed. B.~Coecke, Springer (2011).

\bibitem{EEK98}{T.~Erker, M.~H.~Escard\'{o} and K.~Keimel.}
\newblock The way-below relation of function spaces over semantic domains.
\newblock \emph{Topology and its Applications} \textbf{89}, 61--74 (1998).

\bibitem{Flo09} C.~Flori.
\newblock A topos formulation of consistent histories.
\newblock \textit{Jour.\ Math.\ Phys} {\bf 51} 053527 (2009).

\bibitem{GHK03} G. Gierz, K.H.~Hofmann, K.~Keimel, J.D.~Lawson, M.W.~Mislove, D.S.~Scott.
\newblock \textit{Continuous Lattices and Domains}.
\newblock Encyclopedia of Mathematics and its Applications \textbf{93}, Cambridge University
Press (2003).

\bibitem{HLS09} C.~Heunen, N.P.~Landsman, B.~Spitters.
\newblock A topos for algebraic quantum theory.
\newblock \emph{Comm.\ Math.\ Phys.} \textbf{291}, 63--110 (2009).

\bibitem{HLS09b} C.~Heunen, N.P.~Landsman, B.~Spitters.
\newblock The Bohrification of operator algebras and quantum logic.
\newblock \emph{Synthese}, in press (2011).

\bibitem{HLS11} C.~Heunen, N.P.~Landsman, B.~Spitters.
\newblock Bohrification.
\newblock In \textit{Deep Beauty}, ed. H.~Halvorson, 271--313, Cambridge University Press, New York (2011).

\bibitem{Ish10} C.J.~Isham.
\newblock Topos methods in the foundations of physics.
In \textit{Deep Beauty}, ed. Hans Halvorson, 187--205, Cambridge University Press, New York (2011).

\bibitem{Ish97} C.J. Isham.
\newblock Topos theory and consistent histories: The internal
logic of the set of all consistent sets.
\newblock {\em Int. J. Theor. Phys.}, \textbf{36}, 785--814 (1997).

\bibitem{IB98} C.J.~Isham and J.~Butterfield.
\newblock A topos perspective on the {K}ochen-{S}pecker theorem:
            {I.} {Q}uantum states as generalised
valuations. \newblock {\em Int.\ J.\ Theor.\ Phys.} \textbf{37},
2669--2733 (1998).

\bibitem{IB99} C.J.~Isham and J.~Butterfield.
\newblock A topos perspective on the {K}ochen-{S}pecker theorem:
    {II.} {C}onceptual aspects, and classical analogues.
\newblock  {\em Int.\ J.\ Theor.\ Phys.} \textbf{38}, 827--859 (1999).

\bibitem{IB00} C.J.~Isham, J.~Hamilton and J.~Butterfield.
\newblock A topos perspective on the {K}ochen-{S}pecker theorem:
    {III.} {V}on {N}eumann algebras as the base category.
\newblock  {\em Int.\ J.\ Theor.\ Phys.} \textbf{39}, 1413-1436 (2000).

\bibitem {IB00b} C.J.~Isham and J.~Butterfield.
\newblock Some possible roles for topos theory in quantum
theory and quantum gravity.
\newblock {\em Found.\ Phys.} {\bf 30}, 1707--1735 (2000).

\bibitem{IB02} C.J.~Isham and J.~Butterfield.
\newblock {A topos perspective on the Kochen-Specker theorem:
        {IV.} {I}nterval valuations}.
\newblock {\em Int.\ J.\ Theor.\ Phys} {\bf 41}, 613--639 (2002).

\bibitem{MP06} K.~Martin, P.~Panangaden.
\newblock A Domain of Spacetime Intervals in General Relativity.
\newblock \textit{Commun. Math. Phys.} \textbf{267}, 563–-586 (2006).

\bibitem{Sco70} D.S.~Scott.
\newblock Outline of a mathematical theory of computation.
\newblock In Proceedings of \textit{4th Annual Princeton Conference on Information Sciences and Systems}, 169--176 (1970).

\bibitem{Sco72} D.S.~Scott.
\newblock Continuous lattices.
\newblock In \textit{Toposes, Algebraic Geometry and Logic}, Lecture Notes in Mathematics \textbf{274}, Springer, 93--136 (1972).

\bibitem{Sco72b} D.S.~Scott.
\newblock Formal semantics of programming languages.
\newblock In \textit{Lattice theory, data types and semantics}, Englewood Cliffs, Prentice-Hall, 66--106 (1972).

\bibitem{Wil70} S.~Willard.
\newblock \textit{General Topology}.
\newblock Addison-Wesley Series in Mathematics, Addison Wesley (1970).

\bibitem{YJ96}{L.~Ying-Ming,  L.~Ji-Hua.}
\newblock Solutions  to  two  problems  of  J.D.  Lawson  and  M.  Mislove.
\newblock \emph{Topology and its Applications} \textbf{69}, 153--164 (1996).

\end{thebibliography}
\end{document}